\documentclass[11pt]{article}
\usepackage{cite}
\usepackage{amssymb}
\usepackage{amsfonts}
\usepackage{mathrsfs}
\usepackage{graphicx}
\usepackage{amsmath}
\usepackage{amsthm}
\usepackage{epstopdf}
\usepackage{booktabs}
\usepackage{pifont,bm}
\usepackage{tikz}
\usepackage[title]{appendix}
\usepackage{color}
\usepackage{enumerate}
\usepackage{threeparttable}
\usepackage{diagbox}

\newtheorem{thm}{Theorem}[section]

\newtheorem{prop}[thm]{Proposition}
\theoremstyle{definition}

\newtheorem{assum}[thm]{Assumption}
\newtheorem{rhp}[thm]{Riemann-Hilbert Problem}

\makeatletter
\def\@biblabel#1{[#1]}
\makeatother

\makeatletter \@addtoreset{equation}{section}

\begin{document}

\begin{titlepage}
\title{\bf{Long time and Painlev\'{e}-type asymptotics  for the defocusing Hirota
equation with finite density initial data
\footnote{
Corresponding authors.\protect\\
\hspace*{3ex} E-mail addresses: ychen@sei.ecnu.edu.cn (Y. Chen)}
}}
\author{Wei-Qi Peng$^{a}$, Yong Chen$^{a,b,*}$\\
\small \emph{$^{a}$School of Mathematical Sciences, Shanghai Key Laboratory of PMMP} \\
\small \emph{East China Normal University, Shanghai, 200241, China} \\
\small \emph{$^{b}$College of Mathematics and Systems Science, Shandong University }\\
\small \emph{of Science and Technology, Qingdao, 266590, China} \\
\date{}}
\thispagestyle{empty}
\end{titlepage}
\maketitle

\vspace{-0.5cm}
\begin{center}
\rule{15cm}{1pt}\vspace{0.3cm}

\parbox{15cm}{\small
{\bf Abstract}\\
\hspace{0.5cm}  In this work, we consider the  Cauchy problem for the defocusing Hirota equation with a nonzero background
\begin{align}
\begin{cases}
iq_{t}+\alpha\left[q_{xx}-2\left(\left\vert q\right\vert^{2}-1\right)q\right]+i\beta\left(q_{xxx}-6\left\vert q\right\vert^{2}q_{x}\right)=0,\quad (x,t)\in \mathbb{R}\times(0,+\infty),\\
q(x,0)=q_{0}(x),\qquad \underset{x\rightarrow\pm\infty 1}{\lim} q_{0}(x)=\pm 1, \qquad q_{0}\mp 1\in H^{4,4}(\mathbb{R}).
\end{cases}
\nonumber
\end{align}
 According to the Riemann-Hilbert problem representation of the Cauchy problem and the $\bar{\partial}$ generalization of the nonlinear steepest descent method, we find  different  long time asymptotics types for the defocusing Hirota
equation in oscillating region and transition region,  respectively.
For the oscillating region $\xi<-8$, four phase points appear on the jump contour $\mathbb{R}$, which arrives at  an asymptotic expansion,given by
\begin{align}
q(x,t)=-1+t^{-1/2}h+O(t^{-3/4}).\nonumber
\end{align}
It consists of three terms. The first term $-1$ is leading term representing a nonzero background, the second term $t^{-1/2}h$ originates from the continuous spectrum
and the third term $O(t^{-3/4})$ is the error term  due to pure $\bar{\partial}$-RH problem.
For the transition region $\vert\xi+8\vert t^{2/3}<C$, three phase points raise on the jump contour $\mathbb{R}$. Painlev\'{e} asymptotics expansion is obtained
\begin{align}
q(x,t)=-1-(\frac{15}{4}t)^{-1/3}\varrho+O(t^{-1/2}),\nonumber
\end{align}
in which the leading term is a solution to the Painlev\'{e} II equation, the last term is
a residual error being from pure $\bar{\partial}$-RH problem and parabolic cylinder model.
}

\vspace{0.5cm}
\parbox{15cm}{\small{

\vspace{0.3cm} \emph{Key words:} Defocusing Hirota equation; $\bar{\partial}$ steepest descent method; Long-time asymptotics; Painlev\'{e} asymptotic.\\

 } }
\end{center}
\vspace{0.3cm} \rule{15cm}{1pt} \vspace{0.2cm}

\section{Introduction}
The classis nonlinear Schr\"{o}dger (NLS) equation has a key action in describing the propagation of a picosecond optical pulse. However, the more complex nonlinear wave phenomena often require the description of the high order NLS equation. The Hirota equation \cite{chen-CNSNS34}, as an important higher-order NLS equations, contains higher-order dispersive and nonlinear effects such as third-order dispersion, self-frequency shift, and self-steepening arising from the
stimulated Raman scattering. Thus, it can be regarded as a more accurate model to study ultra-short optical pulse propagation \cite{chen-CNSNS6,chen-CNSNS31,chen-CNSNS32,chen-CNSNS33}.
In this paper, we investigate the long time and Painlev\'{e}-type asymptotics behavior for the Cauchy problem
of the defocusing Hirota equation
under nonzero boundary conditions(NZBCs)
\begin{align}\label{1}
\begin{cases}
iq_{t}+\alpha\left[q_{xx}-2\left(\left\vert q\right\vert^{2}-1\right)q\right]+i\beta\left(q_{xxx}-6\left\vert q\right\vert^{2}q_{x}\right)=0,\quad (x,t)\in \mathbb{R}\times(0,+\infty),\\
q(x,0)=q_{0}(x),\qquad \underset{x\rightarrow\pm\infty 1}{\lim} q_{0}(x)=\pm 1, \qquad q_{0}\mp 1\in H^{4,4}(\mathbb{R}),
\end{cases}
\end{align}
where the weighted Sobolev space $H^{4,4}(\mathbb{R})$ is defined as
\begin{align}
&H^{4,4}(\mathbb{R})=L^{2,4}(\mathbb{R})\cap H^{4}(\mathbb{R}),\notag\\
&L^{2,4}(\mathbb{R})=\left\{\left(1+\left\vert \cdot\right\vert^{2}\right)^{2}f\in
L^{2}(\mathbb{R})\right\},\notag\\
&H^{4}(\mathbb{R})=\left\{f\in L^{2}(\mathbb{R})|\partial^{j}f\in L^{2}(\mathbb{R}), j=0,\cdots,4\right\}.\nonumber
\end{align}
Real constants $\alpha$ and $\beta$ represent the second-order and third-order dispersions, respectively. When $\alpha=1, \beta=0$, the defocusing Hirota equation with NZBCs can reduce to the defocusing NLS equation with NZBCs. While at the case of $\alpha=0, \beta=1$, the defocusing Hirota equation with NZBCs reduces to the defocusing modified Korteweg-de Vries (mKdV) equation with NZBCs. Recently, the dark soliton solutions of the defocusing Hirota equation have been studied via the binary  Darboux transform \cite{ZhangHQ}.  Soliton solutions for defocusing Hirota equation with NZBCs were obtained using the inverse scattering transforms \cite{chen-CNSNS} and $\bar{\partial}$-dressing method \cite{Huang-ND}. $N$-double poles solutions for nonlocal Hirota equation with NZBCs were derived using Riemann-Hilbert(RH) method and PINN algorithm \cite{Peng-PD}. Conservation laws of the defocusing Hirota equation under NZBCs was studied in \cite{Xu}. Painlev\'{e}-type asymptotics for the defocusing Hirota equation with zero
boundary conditions(ZBCs) in transition region was discussed by classical nonlinear steepest descent method \cite{Xun}.

For integrable systems, the inverse scattering transform is an important milestone in studying the
Cauchy problem for the nonlinear evolution equations. The inverse scattering transform based on the RH problem is to use the solution of a matrix RH problem to represent the solution of the corresponding equation. The existence of global solutions on the
line was also studied from the perspective of inverse scattering transform based on the representation of a RH problem \cite{Pelinovsky1,Pelinovsky2}. And what's more, undergoes decades of development, RH method is extended to derive long time asymptotics  of the solution, which is an important application. The study of long time asymptotics can first be traced back to the work of Zakharov and Manakov \cite{Guo-JDE20}, they analysed  the long time asymptotics  for the  NLS equation. Stimulated by this work,
the famous nonlinear steepest descent method was raised by
Deift and Zhou  to solve the initial value
problems for the mKdV equation with a oscillatory RH problem \cite{Guo-JDE6}. Subsequently,  a good deal of original
significant results about long time asymptotics of solutions to various integrable systems
were researched, including the defocusing
NLS equation \cite{Peng-17},  the KdV equation \cite{Peng-18}, the sine-Gordon equation \cite{Peng-19}, the Cammasa-Holm equation \cite{Peng-20}, Fokas-Lenells equation \cite{Peng-21}, the Hirota equation \cite{Peng-22}, the Kundu-Eckhaus equation \cite{Peng-23} et al. Nevertheless, the classical nonlinear steepest descent method needs to consider the fine $L_{p}$ estimates of Cauchy
projection operators, to avoid this matter, the $\bar{\partial}$-steepest descent method was advocated  by McLaughlin and Miller \cite{Tian-wang24,Tian-wang25}. Since then, this method has gained wide attention, more and more progress has been made, including defocusing NLS equation \cite{Tian-wang26,Tian-wang27},  focusing NLS equation \cite{Tian-wang29}, derivative NLS equation \cite{Liujiaqi,Liujiaqi1}, KdV equation \cite{Tian-wang28}, short pulse equation \cite{Tian-wang30}, focusing Fokas-Lenells equation \cite{Tian-wang31}, modified Camassa-Holm equation \cite{Tian-wang32}, WKI eqaution \cite{Tian-wang33} et al.

In recent years,  $\bar{\partial}$-steepest descent method also has extended to study long time asymptotics for the integrable
systems with NZBCs, such as defocusing NLS equation \cite{Wang-CMP}, Novikov equation \cite{Yang-CNS},  defocusing mKdV equation \cite{Xu-Fan}, WKI equation \cite{Tian-wang34}, modified Camassa-Holm equation \cite{Tian-wang35} et al.
For the Painlev\'{e}-type asymptotics in the field of integrable system, there are also a lot of work having been done. For example,
the Painlev\'{e}-type asymptotics for the KdV equation were firstly discussed by Segur and Ablowitz \cite{Wang-KDV5}. After that, the asymptotic relationship for the mKdV equation between different regions was given by Deift and Zhou  \cite{Guo-JDE6}. Applying the nonlinear steepest descent technique, Monvel has developed  the Painlev\'{e}-type asymptotics for the Camassa-Holm
equation \cite{Wang-KDV14}. Worth mentioning, Charlier and Lenells established
the Airy and Painlev\'{e} asymptotics for the mKdV equation \cite{Wang-KDV15}. Also, Painlev\'{e}-type asymptotics of an extended mKdV equation in transition regions has been obtain \cite{Liunan}. Soon afterwards, the Painlev\'{e}-type
asymptotics were generalized to the whole mKdV hierarchy \cite{Wang-KDV16}. Very recently, using the $\bar{\partial}$-steepest descent method,  the
Painlev\'{e}-type asymptotics for defocusing NLS equation under NZBCs were creatively
raised by Wang and Fan \cite{Wang-KDV19}, and then they also proposed the Painlev\'{e}-type asymptotics of defocusing mKdV equation with NZBCs \cite{Wang-KDV20}.

In this article, we analyse the situation of solitonless on the long time asymptotics for the defocusing Hirota equation with NZBCs. We extend the $\bar{\partial}$-steepest descent method to the defocusing Hirota equation with NZBCs and
derive  different forms of long time asymptotics in two different region. To the best of our knowledge,  the long time asymptotics studies of the defocusing Hirota equation with NZBCs have been
rarely reported via using $\bar{\partial}$-steepest descent method, because the $\bar{\partial}$-steepest descent method
with NZBCs is more complicated than one with ZBCs. Compared to the NLS equation, the defocusing Hirota equation has a more complex phase point distribution, which leads to that the construction of the mixed $\bar{\partial}$-RH problem is more difficult and tricky. Besides, for the  transition region, we need consider three phase points instead of one compared to the NLS equation, and two of these phase points need to be matched by the parabolic cylinder equation. The main results of this paper
is generalized in what follows:\\
\textbf{Main result:}
\begin{thm}\label{thm1}
Suppose that $q(x,t)$ be the solution of the Cauchy problem \eqref{1} with the initial data
$q_{0}$  lying in  $\tanh(x)+H^{4,4}(\mathbb{R})$. Then, at the region $\xi<-8, \xi=O(1)$,  for any constant $T=T(q_{0},\xi)$, the long time asymptotics of the solution $q(x, t)$ is
\begin{align}
q(x,t)=-1+t^{-1/2}h(x,t)+O(t^{-3/4}), \qquad t>T,\nonumber
\end{align}
where
\begin{align}
h(x,t)=\sum_{k=1}^{4}\frac{\sqrt{2\pi}e^{-\frac{\pi}{4}i\epsilon_{k}-\frac{\pi\nu(\xi_{k})}{2}}}{\sqrt{2\epsilon_{k}f''(\xi_{k})}
r_{\xi_{k}}\Gamma(i\epsilon_{k}\nu(\xi_{k}))},\nonumber
\end{align}
with
\begin{align}
r_{\xi_{k}}=r(\xi_{k})\delta_{0}^{2}(\xi_{k})e^{-2itf(\xi_{k})-i\epsilon_{k}\nu(\xi_{k})\ln(2t\epsilon_{k}f''(\xi_{k}))}, \ \nu(\xi_{k})=\frac{1}{2\pi}\ln(1-\left\vert r(\xi_{k})\right\vert^{2}),\ \epsilon_{k}=(-1)^{k+1},\nonumber
\end{align}
and $\Gamma$  is Euler's Gamma function.
\end{thm}
\begin{thm}\label{thm2}
Suppose the reflection
coefficient $\{r(z), R(z)\}$ and the discrete spectrum $\{z_{j}\}_{j=1}^{n}$ being both associated with the
initial data $q_{0}$ lying in weighted Sobolev space $\tanh(x)+H^{4,4}(\mathbb{R})$. Then the Painlev\'{e} asymptotics of the
solution to the Cauchy problem \eqref{1} in transition region $\vert\xi+8\vert t^{2/3}<C$ is
\begin{align}
q(x,t)=-1-(\frac{15}{4}t)^{-\frac{1}{3}}\varrho+O(t^{-\frac{1}{2}}),\nonumber
\end{align}
where
\begin{align}
\varrho=\frac{i}{2}\left(\int_{s}^{\infty}u^{2}(\zeta)d\zeta+e^{i\varphi_{0}}u(s)\right),\nonumber
\end{align}
with
\begin{align}
s=\frac{4}{15}(8+\xi)(\frac{15}{4}t)^{2/3},\qquad  \varphi_{0}=arg R(-1),\nonumber
\end{align}
and $u(s)$ is a solution of the following Painlev\'{e} II equation
\begin{align}
u_{ss}=2u^{3}+su,\nonumber
\end{align}
whose asymptotics as $s\rightarrow\infty$ is
\begin{align}
u(s)\sim -\vert r(-1)\vert Ai(s),\nonumber
\end{align}
where $Ai(s)$ is the classical Airy function.
\end{thm}

\textbf{Organization of the Rest of the Work}:
In section 2, based on the Lax pair of the defocusing Hirota equation under NZBCs, we introduce the
analyticity, symmetries and asymptotic properties for the eigenfunctions and scattering data,
and the basis RH problem $M(z)$ is constructed
for the Cauchy problem of the defocusing Hirota equation.
In section 3, we introduce the matrix function $\delta(z)$ and interpolation function $G(z)$ to generate the new RH problem
$M^{(1)}(z)$ whose jump matrix need to be decomposed into two triangle matrices near the phase points $\xi=\xi_{k},k=1\cdots,4$. Next, through introducing a matrix function $R^{(2)}(z)$, a continuous extension of the jump matrix off the
real axis  is constructed and the mixed $\bar{\partial}$-RH problem was derived.
The mixed $\bar{\partial}$-RH problem can be decomposed into two parts including
a model RH problem with $\bar{\partial}R^{(2)}=0$ i.e., $M^{RHP}$  and a pure $\bar{\partial}$-RH problem with $\bar{\partial}R^{(2)}\neq 0$
i.e., $M^{(4)}$. Of which the model RH problem $M^{RHP}$ can be solved near
the phase point $\xi_{k}$ by matching parabolic cylinder model
problem. Besides, the error function $E(z)$ with a small-norm RH problem is
produced. Then, the error estimation is given via analysing
the pure $\bar{\partial}$-RH problem $M^{(4)}$.
Finally, we obtain the long time asymptotic  behavior of the defocusing Hirota equation under NZBCs in the oscillating region $\xi<-8$. In section 4, using the analogous $\bar{\partial}$ technique as section 3, through deriving a
solvable RH model that matches with the Painlev\'{e} model, we obtain the Painlev\'{e}-type
asymptotics in the  transition region $\vert\xi+8\vert t^{2/3}<C$.

\section{Spectral analysis and basic RH problem}
In this section, we briefly review the direct scattering component, see \cite{chen-CNSNS} for details.
\subsection{Spectral analysis on Lax pair}
The defocusing Hirota equation \eqref{1} satisfies the following Lax pair
\begin{align}\label{2}
\Phi_{x}=X\Phi=(ik\sigma_{3}+Q)\Phi, \qquad \Phi_{t}=T\Phi=(\alpha T_{NLS}+\beta T_{CMKdV})\Phi,
\end{align}
where
\begin{align}
&T_{NLS}=-2kX+i\sigma_{3}(Q_{x}-Q^{2}+1),\notag\\
&T_{CMKdV}=-2k(T_{NLS}-i\sigma_{3})+[Q_{x},Q]+2Q^{3}-Q_{xx},\nonumber
\end{align}
with $k\in \mathbb{C}$ being a spectral parameter, and
\begin{align}
\sigma_{3}=\left(\begin{array}{cc}
    1  &  0\\
    0  &  -1\\
\end{array}\right),\quad
Q=Q(x,t)=\left(\begin{array}{cc}
    0  &  q\\
    \bar{q}  &  0\\
\end{array}\right),\nonumber
\end{align}
of which the $\overline{D}$ on behalf of the complex conjugate of $D$.
Under the boundary condition of $q$, the Lax pair \eqref{2} turns into
\begin{align}\label{5}
\Phi_{\pm,x}\sim X_{\pm}\Phi_{\pm}, \qquad \Phi_{\pm,t}\sim T_{\pm}\Phi_{\pm},\qquad x\rightarrow \pm\infty,
\end{align}
where
\begin{align}
X_{\pm}=ik\sigma_{3}+Q_{\pm},\quad T_{\pm}=\left[\beta\left(4k^{2}+2\right)-2\alpha k\right]X_{\pm},\nonumber
\end{align}
and
\begin{align}
Q_{\pm}=\left(\begin{array}{cc}
    0  &  \pm 1\\
    \pm 1  &  0\\
\end{array}\right).\nonumber
\end{align}
The eigenvalues of the $X_{\pm}$ are $\pm i\lambda$, which admit
\begin{align}
\lambda^{2}=k^{2}-1.\nonumber
\end{align}
To circumvent the multi-valued nature of  eigenvalue $\lambda$, an uniformization variable is  introduced as follows
\begin{align}
z=k+\lambda,\nonumber
\end{align}
then two single-valued functions can be expressed as
\begin{align}
k=\frac{1}{2}(z+\frac{1}{z}),\quad \lambda=\frac{1}{2}(z-\frac{1}{z}).\nonumber
\end{align}
As as result, the asymptotic spectral problem \eqref{5} satisfies the following Jost solution
\begin{align}
\Phi_{\pm}(z)\sim \begin{cases}
 E_{\pm}e^{i\theta\sigma_{3}},\ k\neq\pm 1, \\
 I+\left[x+(6\beta-2\alpha k)t\right]X_{\pm}, \ k=\pm 1,
 \end{cases}\nonumber
\end{align}
where
\begin{align}
E_{\pm}=\left(\begin{array}{cc}
    1  &  \frac{\pm i}{z}\\
    \frac{\mp i}{z}  &  1\\
\end{array}\right),\ \theta(x,t, z)=\lambda\left\{x+\left[\beta\left(4k^{2}+2\right)-2\alpha k\right]t\right\}.\nonumber
\end{align}
Making a gauge transformation
\begin{align}
\mu_{\pm}(x,t, z)=\Phi_{\pm}(x,t, z)e^{-i\theta\sigma_{3}},\nonumber
\end{align}
then as $x\rightarrow \pm \infty$, one has
\begin{align}
\mu_{\pm}(z)\sim E_{\pm}(z),\ \det(\Phi_{\pm})=\det(\mu_{\pm})=1-\frac{1}{z^{2}}.\nonumber
\end{align}
$\mu_{\pm}(z)$ meet the following Volterra type integral equations
\begin{align}\label{15}
\mu_{\pm}(z)=\begin{cases}
E_{\pm}+E_{\pm}\int_{\pm \infty}^{x}e^{i\lambda(x-y)\hat{\sigma}_{3}}\left[E_{\pm}^{-1}\Delta Q_{\pm}\mu_{\pm}\right]\mathrm{d}y, \ z\neq \pm 1,\\
E_{\pm}+\int_{\pm \infty}^{x}\left[I+(x-y)(Q_{\pm}\pm i\sigma_{3})\right]\Delta Q_{\pm}\mu_{\pm}\mathrm{d}y, \ z= \pm 1.
\end{cases}
\end{align}
where $e^{\hat{\sigma}_{3}}D=e^{\sigma_{3}}De^{-\sigma_{3}}, \Delta Q_{\pm}=Q-Q_{\pm}$. Then, in terms of the definition of $\mu(z)$ and the above integrals \eqref{15},
we can derive the analytical properties of the eigenfunctions $\mu_{\pm}(z)$.
\begin{prop}
 {\rm (Analytic property)}Assuming that $q(x)-q_{0}\in L^{1,n+1}(\mathbb{R})$ and $q'\in W^{1,1}(\mathbb{R})$,
we find that $\mu_{+,1}, \mu_{-,2}$ are analytic in $\mathbb{C}_{+}$ and $\mu_{-,1}, \mu_{+,2}$ are analytic in $\mathbb{C}_{-}$. The
$\mu_{\pm,j}(j=1, 2)$ denote the $j$-th column of $\mu_{\pm}$, and $\mathbb{C}_{\pm}$ denote the upper and lower complex $z$-plane,
respectively.
\end{prop}

\begin{prop}
 {\rm (Symmetry property)} The eigenfunctions $\mu_{\pm}(x, t; \lambda)$ possess
the following symmetry relation
\begin{align}
\mu_{\pm}(x, t,z)=\sigma_{1}\overline{\mu_{\pm}(x, t, \overline{z})}\sigma_{1}=\frac{i}{z}\mu_{\pm}(x, t, \frac{1}{z})\sigma_{3}Q_{\pm},\nonumber
\end{align}
\end{prop}

\begin{prop}
 {\rm (Asymptotic property)} The eigenfunctions $\mu_{\pm}(x, t, z)$ possess the following asymptotic behavior
\begin{align}
\mu_{\pm}(x, t, z)=\mathbb{I}+O(\frac{1}{z}), \ z\rightarrow\infty,\quad \mu_{\pm}(x, t, z)=\frac{i}{z}\sigma_{3}Q_{\pm}+O(1), \ z\rightarrow 0.\nonumber
\end{align}
\end{prop}

For $z\in \mathbb{R}$ both eigenfunctions $\mu_{+}(x, t, z)$ and $\mu_{-}(x, t, z)$ are the fundamental
matrix solutions, there exists a scattering matrix $S(z)$ arriving at
\begin{align}
\mu_{+}(x, t, z)=\mu_{-}(x, t, z)e^{i\theta\hat{\sigma_{3}}}S(z), \quad z\in R\setminus \{0,\pm 1\},\nonumber
\end{align}
where the matrix $S(z)=(s_{ij}(z))_{2\times 2}$.
The functions $s_{11}(z), s_{22}(z)$ can be written as
\begin{align}
s_{11}(z)= \frac{\det(\Phi_{+1},\Phi_{-2})}{1-z^{-2}},\quad  s_{22}(z)= \frac{\det(\Phi_{-1},\Phi_{+2})}{1-z^{-2}},\nonumber
\end{align}
where $\det (D)$ means the determinate of a matrix $D$. It is not hard to find that
$s_{11}(z)$ is analytic in $\mathbb{C}_{+}$, $s_{22}(z)$ is analytic in $\mathbb{C}_{-}$.

\begin{prop}
 {\rm (Analytic property)} Let $z\in \mathbb{R}\setminus \{0,\pm 1\}$,
then, $s_{11}(z)$ are analytic in $\mathbb{C}_{+}$ and has no singularity on the real axis $\mathbb{R}$.
 $s_{22}(z)$ are analytic in $\mathbb{C}_{-}$ and has no singularity on the real axis $\mathbb{R}$.
\end{prop}

\begin{prop}
 {\rm (Symmetry property)} The scattering data $s_{ij}(z)$ possess
the following symmetry relation
\begin{align}
s_{11}(z)=\overline{s_{22}(\bar{z})}=-s_{22}(\frac{1}{z}), \quad s_{12}(z)=\overline{s_{21}(\bar{z})}=s_{21}(\frac{1}{z}).\nonumber
\end{align}
\end{prop}

\begin{prop}
 {\rm (Asymptotic property)} The scattering data  $s_{11}(z)$ possess the following asymptotic behavior
\begin{align}
s_{11}(z)=1+O(\frac{1}{z}), \quad z\rightarrow\infty, \qquad s_{11}(z)=-1+O(z), \quad z\rightarrow 0.\nonumber
\end{align}
\end{prop}

\begin{assum}\label{Ass1}
In this paper, we assume that the initial value $q_{0}(x)$
is selected such that $s_{11}(z)$ has finite $N$ simple zeros $z_{1},z_{2},\cdots, z_{N}$ on $\mathbb{C}_{+}\cap \{z: \vert z\vert=1\}$, and they distribute in the unit circle. Zeros of $s_{22}(z)$ is corresponding in $\mathbb{C}_{-}$.
\end{assum}

The symmetries of $s_{ij}(z)$ declare that $s_{11}(z_{n})=0\Leftrightarrow \overline{s_{22}(\overline{z_{n}})}=0$, therefore the discrete spectrum can given by
\begin{align}
\Upsilon=\left\{z_{n},\overline{z_{n}}\right\}_{n=1}^{N}.\nonumber
\end{align}
And the distribution of $\Upsilon$ on the $z$-plane is displayed in Fig. 1.\\
\centerline{\begin{tikzpicture}[scale=1.5]
\path [fill=gray] (2.5,0) -- (0.5,0) to
(0.5,2) -- (2.5,2);
\path [fill=gray] (4.5,0) -- (2.5,0) to
(2.5,2) -- (4.5,2);
\draw[-][thick](0.5,0)--(0.75,0);
\draw[-][thick](0.75,0)--(1,0);
\draw[->][thick](1,0)--(2,0);
\draw[-][thick](2,0)--(2.5,0);
\draw[fill] (2.5,0) circle [radius=0.03];
\draw[->][thick](2.5,0)--(3,0);
\draw[-][thick](3,0)--(4,0);
\draw[-][thick](4,0)--(4.5,0)node[above]{$\mbox{Re}z$};
\draw[-][thick](2.5,2)node[right]{$\mbox{Im}z$}--(2.5,0);
\draw[-][thick](2.5,0)--(2.5,-2);
\draw[-][thick](2.5,-1.5)--(2.5,-0.5);
\draw[-][thick](2.5,-2)--(2.5,-1.5);
\draw[-][thick](2.5,1.5)--(2.5,0.5);
\draw[-][thick](2.5,2)--(2.5,1.5);
\draw[fill] (2.5,-0.1) node[right]{$0$};
\draw[fill] (1.5,0) circle [radius=0.03];
\draw[fill] (1.3,0) node[below]{$-1$};
\draw[fill] (3.5,0) circle [radius=0.03];
\draw[fill] (3.6,0) node[below]{$1$};
\draw[fill](3.2,0.7) circle [radius=0.03] node[right]{$z_{n}$};
\draw[fill] (3.2,-0.7) circle [radius=0.03] node[left]{$\overline{z_{n}}$};
\draw[-][thick](3.5,0) arc(0:360:1);
\draw[-][thick](3.5,0) arc(0:30:1);
\draw[-][thick](3.5,0) arc(0:150:1);
\draw[-][thick](3.5,0) arc(0:210:1);
\draw[-][thick](3.5,0) arc(0:330:1);
\end{tikzpicture}}
\noindent {\small \textbf{Figure 1.} (Color online) The discrete spectrums distribute on the unite circle $\{z: \vert z\vert=1\}$ on the $z$-plane,Region $\mathbb{C}_{+}=\left\{z\in \mathbb{C} | \mbox{Im}z> 0\right\}$ (gray region), region $\mathbb{C}_{-}=\left\{z\in \mathbb{C} | \mbox{Im}z< 0\right\}$ (white region).}

\subsection{A RH problem}
In this subsection, based on the Jost solutions and the scattering relation, we construct a RH problem by defining a sectionally meromorphic matrices
\begin{align}
M(x, t, z)=\begin{cases}
M_{+}(x, t, z)=\left(\frac{\mu_{+1}}{s_{11}}, \mu_{-2}\right),\quad z\in \mathbb{C}_{+},\\
M_{-}(x, t, z)=\left(\mu_{-1}, \frac{\mu_{+2}}{s_{22}}\right),\quad z\in \mathbb{C}_{-},
\end{cases}\nonumber
\end{align}
where $M_{\pm}(x, t, z)=\underset{\epsilon\rightarrow 0}{\lim}M_{\pm}(x, t, z\pm i\epsilon), z\in \mathbb{R}$.

For the initial data that admits Assumption \ref{Ass1}, we can obtain the following matrix RH problem.

\begin{rhp}\label{rhp1}
Find an analysis function $M(x, t, z)$ with the
following properties:\\

 $\bullet$ $M(x, t, z)$ is analytical in $\mathbb{C}\setminus (\mathbb{R}\cup \Upsilon)$ and has simple poles in $\Upsilon$;\\

 $\bullet$ $M_{+}(x, t, z)=M_{-}(x, t, z)J(x, t, z), z\in \mathbb{R}$, where
\begin{align}\label{25}
J(x, t, z)=\left(\begin{array}{cc}
    1-\left\vert r(z) \right\vert^{2}  &  -\overline{r(z)}e^{2i\theta(z)}\\
   r(z)e^{-2i\theta(z)}  &  1 \\
\end{array}\right),
\end{align}

 with $r(z)=\frac{s_{21}}{s_{11}}$;

 $\bullet$ Residue condition:
 \begin{align}\label{26}
&\underset{z=z_{n}}{Res}M=\lim_{z\rightarrow z_{n}}M\left(\begin{array}{cc}
    0  &  0\\
    c_{n}e^{-2i\theta(z_{n})}  &  0\\
\end{array}\right),\notag\\
&\underset{z=\overline{z_{n}}}{Res} M=\lim_{z\rightarrow \overline{z_{n}}}M\left(\begin{array}{cc}
    0  &  \overline{c_{n}}e^{2i\theta(\overline{z_{n}})}\\
    0  &  0\\
\end{array}\right),
\end{align}

where $c_{n}=\frac{s_{21}(z_{n})}{s'_{11}(z_{n})}$;

$\bullet$ $M(x, t, z)=\mathbb{I}+O(\frac{1}{z})$ as $z\rightarrow\infty$;\\

$\bullet$ $M(x, t, z)=\frac{\sigma_{2}}{z}+O(1)$ as $z\rightarrow 0$ .
\end{rhp}

\begin{prop}
If $M(x, t, z)$ admits the above conditions, then the solution of RH problem \ref{rhp1}
is unique. Furthermore, in terms of the solution of RH problem
\ref{rhp1}, the solution $q(x, t)$ of the  Cauchy problem \eqref{1} can be
written as
\begin{align}\label{27}
q(x,t)=-\lim_{z\rightarrow \infty}i\left(zM(z)\right)_{12}.
\end{align}
\end{prop}

Besides, the trace formulae of $s_{11}(z)$ is
\begin{align}
s_{11}(z)=\mbox{exp}\left[-\frac{1}{2\pi i}\int_{\mathbb{R}}\frac{\log(1-\vert r(\zeta)\vert^{2})}{\zeta-z}\mathrm{d}\zeta\right] \prod_{n=1}^{2N}\frac{z-z_{n}}{z-\overline{z_{n}}}.\nonumber
\end{align}

\section{Asymptotics in oscillating region}
In this section, based on the associated matrix RH problem \ref{rhp1}, we study the long time
asymptotics of solution $q(x,t)$ for the defocusing Hirota equation \eqref{1} in oscillating region $\xi<-8$.

\subsection{Distribution of Saddle Points and Signature Table}
It is easily to find that the long time asymptotic behavior of RH problem \ref{rhp1} is influenced by the growth
and decay of the exponential function
\begin{align}
e^{\pm 2ift},\quad f=\frac{1}{2}(z-\frac{1}{z})\left[\frac{x}{t}+\beta\left((z+\frac{1}{z})^{2}+2\right)-\alpha (z+\frac{1}{z})\right].\nonumber
\end{align}
Therefore, in order to ensure the exponential
decaying property, it is necessary to analyse the real part of $2itf$,  given by
\begin{align}
Re(2if)=Imz\left[\vert z\vert^{2}-4(Rez)^{2}+2Rez-\xi-3-4(Rez)^{2} \vert z\vert^{-6}+(2Rez+1)\vert z\vert^{-4}-(\xi+3)\vert z\vert^{-2}\right],\nonumber
\end{align}
where $\xi=\frac{x}{t}$. The stationary phase points can be given out by solving the following equation
\begin{align}\label{30}
f'(z)=\frac{1}{2}z^{-1}\left[3\beta\eta^{3}-2\alpha\eta^{2}-(6\beta-\xi)\eta+4\alpha\right]=0,
\end{align}
where $\eta=z+z^{-1}$. The solutions of equation \eqref{30} are
\begin{align}
&\eta_{1}=\frac{2\alpha}{9\beta}+\sqrt[3]{-\frac{q}{2}+\sqrt{(\frac{q}{2})^{2}+(\frac{p}{3})^{3}}}
+\sqrt[3]{-\frac{q}{2}-\sqrt{(\frac{q}{2})^{2}+(\frac{p}{3})^{3}}},\notag\\
&\eta_{2}=\frac{2\alpha}{9\beta}+\omega\sqrt[3]{-\frac{q}{2}+\sqrt{(\frac{q}{2})^{2}+(\frac{p}{3})^{3}}}
+\omega^{2}\sqrt[3]{-\frac{q}{2}-\sqrt{(\frac{q}{2})^{2}+(\frac{p}{3})^{3}}},\notag\\
&\eta_{3}=\frac{2\alpha}{9\beta}+\omega^{2}\sqrt[3]{-\frac{q}{2}+\sqrt{(\frac{q}{2})^{2}+(\frac{p}{3})^{3}}}
+\omega\sqrt[3]{-\frac{q}{2}-\sqrt{(\frac{q}{2})^{2}+(\frac{p}{3})^{3}}},\nonumber
\end{align}
where
\begin{align}
&\omega=\frac{-1+\sqrt{3}i}{2},\quad p=-\frac{4\alpha^{2}}{27\beta^{2}}-\frac{6\beta-\xi}{3\beta},\notag\\
&q=\frac{2\alpha(\beta\xi+12\beta^{2}-\frac{8}{27}\alpha^{2})}{27\beta^{3}}.\nonumber
\end{align}
Without loss of generality, we take $\alpha=\beta=1$ and have the following distributions of phase points:

\begin{prop}
The distributions of phase points:

\textbf{Case 1:} For $\xi<-8$, the four real phase points $\xi_{j}, j=1, 2, 3, 4$ are located on the jump
contour $\Sigma=\mathbb{R}\setminus \{0\}$. Moreover, we have $\xi_{4}< -1 < \xi_{3} < 0 < \xi_{2} < 1 < \xi_{1}$ and
$\xi_{1}=\frac{\eta_{1}+\sqrt{\eta_{1}^{2}-4}}{2},\xi_{2}=\frac{\eta_{1}-\sqrt{\eta_{1}^{2}-4}}{2},\xi_{3}
=\frac{\eta_{2}+\sqrt{\eta_{2}^{2}-4}}{2},\xi_{4}=\frac{\eta_{2}-\sqrt{\eta_{2}^{2}-4}}{2}$;

\textbf{Case 2:} For $\xi=-8$, the three real phase points $\xi_{j}, j=1, 2, 3$ are located on the jump
contour $\Sigma=\mathbb{R}\setminus \{0\}$. Moreover, we have $\xi_{3}=-1 < 0 < \xi_{2} < 1 < \xi_{1}$;

\textbf{Case 3:} For $-8<\xi<-4$, the two real phase points $\xi_{j}, j=1, 2$ are located on the jump
contour $\Sigma=\mathbb{R}\setminus \{0\}$. Moreover, we have $0 < \xi_{2} < 1 < \xi_{1}$;

\textbf{Case 4:} For $\xi=-4$, the one real phase points $\xi_{1}$ are located on the jump
contour $\Sigma=\mathbb{R}\setminus \{0\}$. Moreover, we have $\xi_{1}=1$;

\textbf{Case 5:} For $\xi>-4$, there are no real phase points being located on the jump
contour $\Sigma=\mathbb{R}\setminus \{0\}$.
\end{prop}
In this paper, we just consider the \textbf{Case 1} and \textbf{Case 2}, which are the oscillating region and transition region, respectively. Their decaying regions of $Re(2if)$ are shown in Fig. 2.\\

{\rotatebox{0}{\includegraphics[width=6.0cm,height=6.0cm,angle=0]{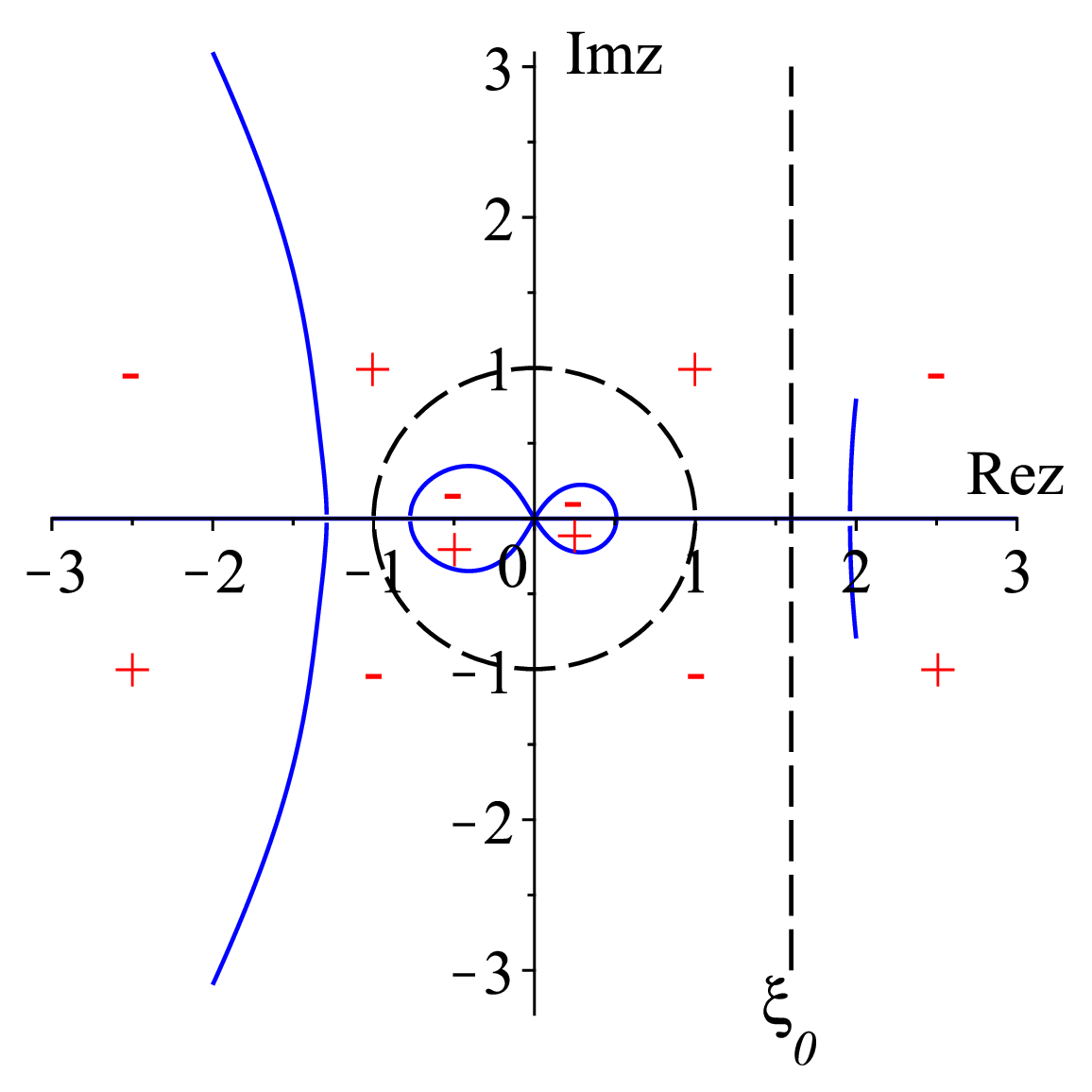}}}
~~
{\rotatebox{0}{\includegraphics[width=6.0cm,height=6.0cm,angle=0]{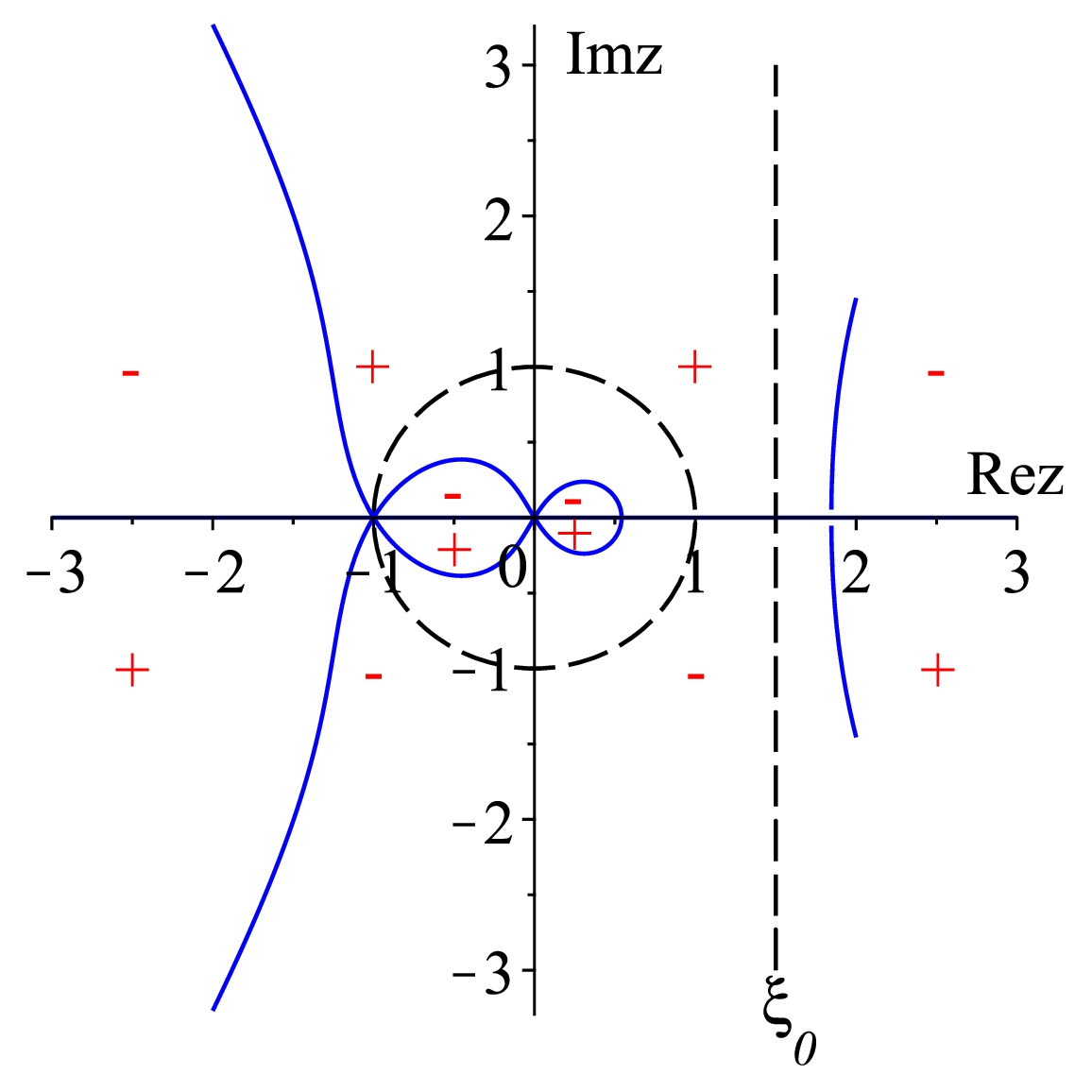}}}\\
$~~~~~~~~~~~~~~~~~~~~~~~(\textbf{a })\xi<-8~~~~
~~~~~~~~~~~~~~~~~~~~~~~~~~~~~~~~~(\textbf{b})\xi=-8$\\
\noindent {\small \textbf{Figure 2.} (Color online) The signature table of $Re(2if(z))$. $``+"$ denotes $Re(2if(z))>0$, $``-"$ denotes $Re(2if(z))<0$.}

\subsection{The first deformation of the basic RH problem}
According to exponential decay Fig. 2, we use following two well known factorizations of the jump matrix $J(z)$ in \eqref{25}:
\begin{align}
J(x, t, z)=
\begin{cases}
\left(\begin{array}{cc}
    1  &  -\overline{r(z)}e^{2itf}\\
    0  &  1\\
\end{array}\right)\left(\begin{array}{cc}
    1  &  0\\
    r(z)e^{-2itf}  &  1\\
\end{array}\right),\ z\in I_{1}, \\
\left(\begin{array}{cc}
    1  &  0\\
    \frac{r(z)e^{-2itf}}{1-\left\vert r(z) \right\vert^{2}}  &  1\\
\end{array}\right)\left(\begin{array}{cc}
    1-\mid r(z) \mid^{2}  &  0\\
    0  &  \frac{1}{1-\mid r(z) \mid^{2}}\\
\end{array}\right)\left(\begin{array}{cc}
    1  &  -\frac{\overline{r(z)}e^{2itf}}{1-\mid r(z) \mid^{2}}\\
    0  &  1\\
\end{array}\right),\ z\in I_{2},
\end{cases}\nonumber
\end{align}
where $I_{1}=(\xi_{4},\xi_{3})\cup(\xi_{2},\xi_{1}), I_{2}=(-\infty,\xi_{4})\cup(\xi_{3},0)\cup(0,\xi_{2})\cup(\xi_{1},+\infty)$. Then  some appropriate sequence of deformations of the
RH problem \ref{rhp1} need to be performed by introducing a scalar function $\delta(z)$:
\begin{align}\label{34}
\delta(z)=\mbox{exp}\left(i\int_{I_{2}}\frac{\nu(\zeta)}{\zeta-z}\mathrm{d}\zeta\right),\
\nu(\zeta)=\frac{1}{2\pi}\ln(1-\left\vert r(\zeta)\right\vert^{2}),
\end{align}
which possesses well properties.

\begin{prop}\label{pjia}
 The function $\delta(z)$ has that:\\

 $\bullet$ $\delta(z)$ is analytical in $\mathbb{C}\setminus I_{2}$;\\

 $\bullet$ For $z\in I_{2}$, the boundary values $\delta(z)$ admits that
\begin{align} \delta_{-}(z)=\delta_{+}(z)(1-\left\vert r(z) \right\vert^{2}),\ z\in I_{2};\nonumber\end{align}

$\bullet$ As $z \rightarrow \infty$, $\delta(z)=1-\frac{i}{z}\int_{I_{2}}\nu(s)\mathrm{d}s+O(z^{-2})$;

$\bullet$ As $z\rightarrow \xi_{k}( k=1,\cdots,4)$ along ray $z=\xi_{k}+e^{i\varphi_{k}}\mathbb{R}_{+}$ with $\left\vert \varphi_{k}\right\vert \leq c < \pi$, one has
\begin{align} \left\vert\delta(z)-\delta_{0}(\xi_{k})(z-\xi_{k})^{\epsilon_{k} i\nu(\xi_{k})}\right\vert\lesssim\left\vert z-\xi_{k}\right\vert^{\frac{1}{2}},\ \epsilon_{k}=(-1)^{k+1},\nonumber \end{align}
where  \begin{align}\delta_{0}(\xi_{k})=\rm{exp}\{i\gamma(\xi_{k},\xi_{k})\},\
&\gamma(z,\xi_{k})=-\epsilon_{k}\nu(\xi_{k})\ln(z-\xi_{k}+1)+\int_{I_{2}}\frac{\nu(\zeta)-\chi(\zeta)\nu(\xi_{k})}{\zeta-z}\mathrm{d}\zeta,\nonumber\end{align}
and $\chi(\zeta)$ is the characteristic function of $\xi_{k}-1,\xi_{k}$.
\end{prop}
Take
\begin{align}\label{38}
\rho<\frac{1}{3}min\{\underset{z_{n}\in \Upsilon}{min}\vert Im z_{n}\vert, \underset{z_{n}, z_{j}\in \Upsilon}{min}\vert z_{n}-z_{j}\vert,\frac{1}{2}\underset{j=1,\cdots,4}{min}\vert \xi_{j}\pm 1\vert,\frac{1}{2}\underset{j=1,\cdots,4}{min}\vert \xi_{j}\vert\}
\end{align}
as the radius of small circles at the center $z_{n}$ or $\overline{z_{n}}$.

Since existence of exponential growth residue conditions \eqref{26}, we further construct the interpolation function to transform the residue conditions into the decreasing jump matrices
\begin{align}
G(z)=
\begin{cases}
\left(\begin{array}{cc}
    1  &  0\\
    \frac{-c_{n}e^{-2ift}}{z-z_{n}}  &  1\\
\end{array}\right),\qquad \vert z-z_{n}\vert<\rho,\\
\left(\begin{array}{cc}
    1  &  \frac{-\overline{c_{n}}e^{2ift}}{z-\overline{z_{n}}}\\
    0  &  1\\
\end{array}\right), \qquad \vert z-\overline{z_{n}}\vert<\rho, \\
\mathbb{I} \qquad elsewhere.
\end{cases}\nonumber
\end{align}

Next, using the scalar function $\delta(z)$ and interpolation function $G(z)$, we transform $M(z)$ into
matrix-value function $M^{(1)}(z)$
\begin{align}\label{41} M^{(1)}(z)=M(z)G(z)\delta^{\sigma_{3}}.\end{align}
Then, the following matrix RH problem for $M^{(1)}(z)$ is:
\begin{rhp}\label{rhp2}
Find an analysis function $M^{(1)}(x, t, z)$ with the
following properties:\\

 $\bullet$ $M^{(1)}(x, t, z)$ is analytical in $\mathbb{C}\setminus \Sigma^{(1)}$ and $\Sigma^{(1)}=\mathbb{R}\cup \Sigma^{pole}$, where $\Sigma^{pole}=\bigcup_{n=1}^{N}\left\{z\in \mathbb{C}: \vert z-z_{n}\vert=\rho \right.$ or $\left. \vert z-\overline{z_{n}}\vert=\rho\right\}$ with a counterclockwise direction;\\

 $\bullet$ $M^{(1)}_{+}(x, t, z)=M^{(1)}_{-}(x, t, z)J^{(1)}(x, t, z), z\in \Sigma^{(1)}$, where
\begin{align}
J^{(1)}(x, t, z)=
\begin{cases}
\left(\begin{array}{cc}
    1  &  -\overline{r(z)}e^{2itf}\delta^{-2}\\
    0  &  1\\
\end{array}\right)\left(\begin{array}{cc}
    1  &  0\\
    r(z)e^{-2itf}\delta^{2}  &  1\\
\end{array}\right),\ z\in I_{1}, \\
\left(\begin{array}{cc}
    1  &  0\\
    \frac{r(z)e^{-2itf}\delta_{-}^{2}}{1-\left\vert r(z) \right\vert^{2}}  &  1\\
\end{array}\right)\left(\begin{array}{cc}
    1  &  -\frac{\overline{r(z)}e^{2itf}\delta_{+}^{-2}}{1-\mid r(z) \mid^{2}}\\
    0  &  1\\
\end{array}\right),\ z\in I_{2},\\
\left(\begin{array}{cc}
    1  &  0\\
    \frac{-c_{n}e^{-2ift}\delta^{2}}{z-z_{n}}  &  1\\
\end{array}\right),\qquad \vert z-z_{n}\vert<\rho,\\
\left(\begin{array}{cc}
    1  &  \frac{-\overline{c_{n}}e^{2ift}\delta^{-2}}{z-\overline{z_{n}}}\\
    0  &  1\\
\end{array}\right), \qquad \vert z-\overline{z_{n}}\vert<\rho;
\end{cases}\nonumber
\end{align}

$\bullet$ $M^{(1)}(x, t, z)=\mathbb{I}+O(\frac{1}{z})$ as $z\rightarrow\infty$;\\

$\bullet$ $M^{(1)}(x, t, z)=\frac{\sigma_{2}}{z}+O(1)$ as $z\rightarrow 0$ .
\end{rhp}

Due to the jump matrices on the circles $\vert z-z_{n}\vert=\rho$ or $\vert z-\overline{z_{n}}\vert=\rho$ exponentially decay
to the identity matrix as $t\rightarrow \infty$, the RH problem \ref{rhp2} is asymptotically
equivalent to the RH problem below.
\begin{rhp}\label{rhp3}
Find an analysis function $M^{(2)}(x, t, z)$ with the
following properties:\\

 $\bullet$ $M^{(2)}(x, t, z)$ is analytical in $\mathbb{C}\setminus \mathbb{R}$;\\

 $\bullet$ $M^{(2)}_{+}(x, t, z)=M^{(2)}_{-}(x, t, z)J^{(2)}(x, t, z), z\in \mathbb{R}$, where
\begin{align}\label{43}
J^{(2)}(x, t, z)=
\begin{cases}
\left(\begin{array}{cc}
    1  &  -\overline{r(z)}e^{2itf}\delta^{-2}\\
    0  &  1\\
\end{array}\right)\left(\begin{array}{cc}
    1  &  0\\
    r(z)e^{-2itf}\delta^{2}  &  1\\
\end{array}\right),\ z\in I_{1}, \\
\left(\begin{array}{cc}
    1  &  0\\
    \frac{r(z)e^{-2itf}\delta_{-}^{2}}{1-\left\vert r(z) \right\vert^{2}}  &  1\\
\end{array}\right)\left(\begin{array}{cc}
    1  &  -\frac{\overline{r(z)}e^{2itf}\delta_{+}^{-2}}{1-\mid r(z) \mid^{2}}\\
    0  &  1\\
\end{array}\right),\ z\in I_{2};
\end{cases}
\end{align}

$\bullet$ $M^{(2)}(x, t, z)=\mathbb{I}+O(\frac{1}{z})$ as $z\rightarrow\infty$;\\

$\bullet$ $M^{(2)}(x, t, z)=\frac{\sigma_{2}}{z}+O(1)$ as $z\rightarrow 0$ .
\end{rhp}

\subsection{The construction of the mixed $\bar{\partial}$-RH problem}
Motivated by the ideas in \cite{Tian-wang24,Tian-wang25,Tian-wang26,Tian-wang27}, we carry out the continuous extensions of the jump matrix off
the real axis, which results  the oscillatory jump into  the decaying jumps. In order to accomplish this purpose and let the opened jump path does not intersect the small disk of any pole, we choose a sufficiently small angle, given by $0<\phi(\xi)<min\{\frac{\pi}{4}, \arctan(\frac{\rho}{1-\vert \xi_{2}(\xi)\vert}), \arctan(\frac{\rho}{1-\vert \xi_{3}(\xi)\vert})\}$. Define
\begin{align}
&\gamma_{41}=\xi_{4}+e^{i(\pi-\phi)}\mathbb{R}_{+},\  \gamma_{44}=\xi_{4}+e^{i\phi}\widetilde{d_{2}}, \ \gamma_{34}=\xi_{3}+e^{i(\pi-\phi)}\widetilde{d_{2}},\notag\\
&\gamma_{31}=\xi_{3}+e^{i\phi}\widetilde{d_{1}}, \ \gamma_{01}=\xi_{0}+e^{i(\pi-\phi)}\widetilde{d_{1}}, \ \gamma_{04}=\xi_{0}+e^{i\phi}d_{1}, \notag\\
&\gamma_{21}=\xi_{2}+e^{i(\pi-\phi)}d_{1},\ \gamma_{24}=\xi_{2}+e^{i\phi}d_{2},\ \gamma_{14}=\xi_{1}+e^{i(\pi-\phi)}d_{2},\notag\\
&\gamma_{11}=\xi_{1}+e^{i\phi}\mathbb{R}_{+}, \  \gamma_{42}=\overline{\gamma_{41}}, \  \gamma_{43}=\overline{\gamma_{44}}, \ \gamma_{33}=\overline{\gamma_{34}}, \ \gamma_{32}=\overline{\gamma_{31}}, \notag\\
&\gamma_{02}=\overline{\gamma_{01}}, \ \gamma_{03}=\overline{\gamma_{04}}, \ \gamma_{22}=\overline{\gamma_{21}}, \ \gamma_{23}=\overline{\gamma_{24}}, \ \gamma_{13}=\overline{\gamma_{14}}, \ \gamma_{12}=\overline{\gamma_{11}},\notag\\
&\gamma_{4+}=\frac{\xi_{4}+\xi_{3}}{2}+e^{\frac{i\pi}{2}}\widetilde{d_{3}},\ \gamma_{3+}=\frac{\xi_{3}}{2}+e^{\frac{i\pi}{2}}\widetilde{d_{4}},\ \gamma_{2+}=\frac{\xi_{2}}{2}+e^{\frac{i\pi}{2}}d_{4},\notag\\
&\gamma_{1+}=\frac{\xi_{2}+\xi_{1}}{2}+e^{\frac{i\pi}{2}}d_{3},\ \gamma_{4-}=\overline{\gamma_{4+}},\ \gamma_{3-}=\overline{\gamma_{3+}},\ \gamma_{2-}=\overline{\gamma_{2+}},\ \gamma_{1-}=\overline{\gamma_{1+}},\nonumber
\end{align}
where $\xi_{0}=0$ and
\begin{align}
&d_{1}\in (0,\frac{\xi_{2}}{2\cos \phi}),\ d_{2}\in (0,\frac{\xi_{1}-\xi_{2}}{2\cos \phi}),\ d_{3}\in (0,\frac{(\xi_{1}-\xi_{2})\tan \phi}{2}),\ d_{4}\in (0,\frac{\xi_{2}\tan \phi}{2}),\notag\\
&\widetilde{d_{1}}\in (0,\frac{\vert\xi_{3}\vert}{2\cos \phi}),\ \widetilde{d_{2}}\in (0,\frac{\xi_{3}-\xi_{4}}{2\cos \phi}),\ \widetilde{d_{3}}\in (0,\frac{(\xi_{3}-\xi_{4})\tan \phi}{2}),\ \widetilde{d_{4}}\in (0,\frac{\vert\xi_{3}\vert\tan \phi}{2}).\nonumber
\end{align}
Then, the complex plane $\mathbb{C}$ is separated into twenty two open sectors $\Omega_{ij}(i=0,\cdots,4,j=1\cdots,4), \Omega_{5}$ and $\Omega_{6}$ see Fig. 3.

\begin{prop}\label{prop1}
There exist functions $R_{kj}\rightarrow \mathbb{C},k=0,\cdots, 4, j=1,\cdots, 4$ such that
\begin{align} &R_{k1}=\begin{cases}
\frac{\overline{r(z)}}{1-\mid r(z)\mid^{2}}\delta_{+}^{-2}(z) \qquad z\in I_{k1},\\
\frac{\overline{r(\xi_{k})}}{1-\mid r(\xi_{k})\mid^{2}}\delta_{0}^{-2}(\xi_{k})(z-\xi_{k})^{-2i\nu(\xi_{k})\epsilon_{k}} \qquad z\in \gamma_{k1},\\
\end{cases}\notag\\
&R_{k2}=\begin{cases}
\frac{r(z)}{1-\mid r(z)\mid^{2}}\delta_{-}^{2}(z) \qquad z\in I_{k2},\\
\frac{r(\xi_{k})}{1-\mid r(\xi_{k})\mid^{2}}\delta_{0}^{2}(\xi_{k})(z-\xi_{k})^{2i\nu(\xi_{k})\epsilon_{k}} \qquad z\in \gamma_{k2},\\
\end{cases}\notag\\
&R_{k3}=\begin{cases}
\overline{r(z)}\delta^{-2}(z) \qquad z\in I_{k3},\\
\overline{r(\xi_{k})}\delta_{0}^{-2}(\xi_{k})(z-\xi_{k})^{-2i\nu(\xi_{k})\epsilon_{k}} \qquad z\in \gamma_{k3},\\
\end{cases}\notag\\
&R_{k4}=\begin{cases}
r(z)\delta^{2}(z) \qquad z\in I_{k4},\\
r(\xi_{k})\delta_{0}^{2}(\xi_{k})(z-\xi_{k})^{2i\nu(\xi_{k})\epsilon_{k}} \qquad z\in \gamma_{k4},\\
\end{cases}\notag\\
&R_{01}=\begin{cases}
\frac{\overline{r(z)}}{1-\mid r(z)\mid^{2}}\delta_{+}^{-2}(z) \qquad z\in I_{01},\\
0 \qquad z\in \gamma_{01},\\
\end{cases},\
R_{02}=\begin{cases}
\frac{r(z)}{1-\mid r(z)\mid^{2}}\delta_{-}^{2}(z) \qquad z\in I_{02},\\
0 \qquad z\in \gamma_{02},\\
\end{cases}\notag\\
&
R_{03}=\begin{cases}
\frac{r(z)}{1-\mid r(z)\mid^{2}}\delta_{-}^{2}(z) \qquad z\in I_{03},\\
0 \qquad z\in \gamma_{03},\\
\end{cases}\ R_{04}=\begin{cases}
\frac{\overline{r(z)}}{1-\mid r(z)\mid^{2}}\delta_{+}^{-2}(z) \qquad z\in I_{04},\\
0 \qquad z\in \gamma_{04},\\
\end{cases}\nonumber
\end{align}
where
\begin{align}
&I_{01}=I_{02}=(\frac{\xi_{3}}{2},0),\ I_{03}=I_{04}=(0,\frac{\xi_{2}}{2}),\  I_{11}=I_{12}=(\xi_{1},+\infty),\   I_{13}=I_{14}=(\frac{\xi_{1}+\xi_{2}}{2},\xi_{1}),\notag\\
&I_{23}=I_{24}=(\xi_{2},\frac{\xi_{1}+\xi_{2}}{2}),\ I_{21}=I_{22}=(\frac{\xi_{2}}{2},\xi_{2}),\ I_{31}=I_{32}=(\xi_{3},\frac{\xi_{3}}{2}),\notag\\
& I_{33}=I_{34}=(\frac{\xi_{3}+\xi_{4}}{2},\xi_{3}),\ I_{43}=I_{44}=(\xi_{4},\frac{\xi_{3}+\xi_{4}}{2}),\  I_{41}=I_{42}=(-\infty,\xi_{4}),\nonumber
\end{align}
and we have the following estimate for $k=0,\cdots, 4$
\begin{align}\label{48} &\left\vert R_{kj}\right\vert\lesssim \sin^{2}\left(arg(z-\xi_{k})\right)+<Re z>^{-1},\notag\\
&\left\vert\bar{\partial} R_{kj}\right\vert\lesssim \left\vert z-\xi_{k}\right\vert^{-1/2}+\left\vert p_{kj}'(Re z)\right\vert,\notag\\
&\left\vert\bar{\partial} R_{kj}\right\vert=0, \qquad \mbox{if} \qquad z\in \Omega_{5}\cup\Omega_{6},\end{align}
where $<Re z>=\sqrt{1+(Re z)^{2}}, p_{k1}= \frac{\overline{r(z)}}{1-\mid r(z)\mid^{2}}, p_{k2}=\frac{r(z)}{1-\mid r(z)\mid^{2}}, p_{k3}=\overline{r(z)}, p_{k4}=r(z)$.
\end{prop}
\centerline{\begin{tikzpicture}[scale=0.8]
\draw[-][thick](5,-2)--(5,-1);
\draw[->][thick](5,0)--(5,-1);
\draw[->][thick](5,2)--(5,1);
\draw[-][thick](5,0)--(5,1);
\draw[-][thick](2,-1)--(2,-0.5);
\draw[->][thick](2,0)--(2,-0.5);
\draw[->][thick](2,1)--(2,0.5);
\draw[-][thick](2,0)--(2,0.5);
\draw[-][thick](0,-1)--(0,-0.5);
\draw[->][thick](0,0)--(0,-0.5);
\draw[->][thick](0,1)--(0,0.5);
\draw[-][thick](0,0)--(0,0.5);
\draw[-][thick](-3,-2)--(-3,-1);
\draw[->][thick](-3,0)--(-3,-1);
\draw[->][thick](-3,2)--(-3,1);
\draw[-][thick](-3,0)--(-3,1);
\draw[<-][thick](-2,-1)--(-3,-2);
\draw[-][thick](-3,-2)--(-4,-1);
\draw[<-][thick](-4,-1)--(-5,0);
\draw[-][thick](-5,0)--(-6,-1);
\draw[<-][thick](-6,-1)--(-7,-2);
\draw[<-][thick](-2,1)--(-3,2);
\draw[-][thick](-3,2)--(-4,1);
\draw[<-][thick](-4,1)--(-5,0);
\draw[-][thick](-5,0)--(-6,1);
\draw[<-][thick](-6,1)--(-7,2);
\draw[-][dashed,thick](-3,0)--(-2,0);
\draw[-][dashed,thick](-7,0)--(-6,0);
\draw[-][dashed,thick](-6,0)--(-5,0);
\draw[-][dashed,thick](-5,0)--(-3,0);
\draw[-][dashed,thick](-2.0,0)--(-1.0,0)node[below]{$\xi_{3}$};
\draw[-][dashed,thick](-1,0)--(0,0);
\draw[->][thick](0,1)--(0.5,0.5);
\draw[-][thick](0.5,0.5)--(1,0);
\draw[->][thick](0,-1)--(0.5,-0.5);
\draw[-][thick](0.5,-0.5)--(1,0);
\draw[-][thick](2,-1)--(1.5,-0.5);
\draw[<-][thick](1.5,-0.5)--(1,0);
\draw[-][thick](2,1)--(1.5,0.5);
\draw[<-][thick](1.5,0.5)--(1,0);
\draw[-][thick](0,1)--(-0.5,0.5);
\draw[->][thick](-1,0)--(-0.5,0.5);
\draw[-][thick](0,-1)--(-0.5,-0.5);
\draw[->][thick](-1,0)--(-0.5,-0.5);
\draw[-][thick](-1,0)--(-1.5,0.5);
\draw[-][thick](-1.5,0.5)--(-2,1);
\draw[-][thick](-1,0)--(-1.5,-0.5);
\draw[-][thick](-1.5,-0.5)--(-2,-1);
\draw[-][dashed,thick](0,0)--(1,0)node[below]{$0$};
\draw[-][dashed,thick](1,0)--(2,0);
\draw[fill] (1,0) circle [radius=0.035];
\draw[fill] (3,0) circle [radius=0.035]node[below]{$\xi_{2}$};
\draw[fill] (7,0) circle [radius=0.035]node[below]{$\xi_{1}$};
\draw[fill] (-5,0) circle [radius=0.035]node[below]{$\xi_{4}$};
\draw[-][dashed,thick](3,0)--(4,0);
\draw[-][dashed,thick](4,0)--(5,0);
\draw[-][dashed,thick](5,0)--(7,0);
\draw[-][dashed,thick](7,0)--(8,0);
\draw[-][dashed,thick](8,0)--(9,0);
\draw[fill] (-1,0) circle [radius=0.035];
\draw[-][dashed,thick](2.0,0)--(3.0,0);
\draw[-][thick](4,1)--(5,2);
\draw[->][thick](5,2)--(6,1);
\draw[-][thick](6,1)--(7,0);
\draw[->][thick](7,0)--(8,1);
\draw[-][thick](8,1)--(9,2);
\draw[-][thick](4,-1)--(5,-2);
\draw[->][thick](5,-2)--(6,-1);
\draw[-][thick](6,-1)--(7,0);
\draw[->][thick](7,0)--(8,-1);
\draw[-][thick](8,-1)--(9,-2);
\draw[fill] [red](1,2.0) node{$\Omega_{5}$};
\draw[fill] [red](1,-2.0) node{$\Omega_{6}$};
\draw[fill] [red](-0.3,0.25) node{$\Omega_{31}$};
\draw[fill] [red](-0.3,-0.25) node{$\Omega_{32}$};
\draw[fill] [red](0.4,0.25) node{$\Omega_{01}$};
\draw[fill] [red](0.4,-0.25) node{$\Omega_{02}$};
\draw[fill] [red](1.65,0.25) node{$\Omega_{04}$};
\draw[fill] [red](1.65,-0.25) node{$\Omega_{03}$};
\draw[fill] [red](2.35,-0.25) node{$\Omega_{22}$};
\draw[fill] [red](2.35,0.25) node{$\Omega_{21}$};
\draw[fill] [red](-2,-0.5) node{$\Omega_{33}$};
\draw[fill] [red](-2,0.5) node{$\Omega_{34}$};
\draw[fill] [red](4,0.5) node{$\Omega_{24}$};
\draw[fill] [red](4,-0.5) node{$\Omega_{23}$};
\draw[fill] [red](-6,0.5) node{$\Omega_{41}$};
\draw[fill] [red](-6,-0.5) node{$\Omega_{42}$};
\draw[fill] [red](-4,0.5) node{$\Omega_{44}$};
\draw[fill] [red](-4,-0.5) node{$\Omega_{43}$};
\draw[fill] [red](6,-0.5) node{$\Omega_{13}$};
\draw[fill] [red](6,0.5) node{$\Omega_{14}$};
\draw[fill] [red](8,0.5) node{$\Omega_{11}$};
\draw[fill] [red](8,-0.5) node{$\Omega_{12}$};
\draw[->][thick](2,1)--(2.5,0.5);
\draw[-][thick](2.5,0.5)--(3,0);
\draw[->][thick](2,-1)--(2.5,-0.5);
\draw[-][thick](2.5,-0.5)--(3,0);
\draw[-][thick](3,0)--(3.5,0.5);
\draw[->][thick](3.5,0.5)--(4,1);
\draw[-][thick](3,0)--(3.5,-0.5);
\draw[->][thick](3.5,-0.5)--(4,-1);
\draw[fill] [cyan](-2,1.3) node{$\gamma_{34}$};
\draw[fill][cyan] (-2,-1.3) node{$\gamma_{33}$};
\draw[fill][cyan] (-0.7,0.7) node{$\gamma_{31}$};
\draw[fill] [cyan](-0.7,-0.7) node{$\gamma_{32}$};
\draw[fill] [cyan](0.7,0.7) node{$\gamma_{01}$};
\draw[fill] [cyan](0.7,-0.7) node{$\gamma_{02}$};
\draw[fill] [cyan](1.3,0.7) node{$\gamma_{04}$};
\draw[fill][cyan] (1.3,-0.7) node{$\gamma_{03}$};
\draw[fill] [cyan](2.7,0.7) node{$\gamma_{21}$};
\draw[fill] [cyan](2.7,-0.7) node{$\gamma_{22}$};
\draw[fill] [cyan](4.0,1.3) node{$\gamma_{24}$};
\draw[fill][cyan] (4.0,-1.3) node{$\gamma_{23}$};
\draw[fill] [cyan](-6,1.4) node{$\gamma_{41}$};
\draw[fill] [cyan](-6,-1.4) node{$\gamma_{42}$};
\draw[fill] [cyan](-4,1.4) node{$\gamma_{44}$};
\draw[fill] [cyan](-4,-1.4) node{$\gamma_{43}$};
\draw[fill] [cyan](-4,1.4) node{$\gamma_{44}$};
\draw[fill] [cyan](6,1.4) node{$\gamma_{14}$};
\draw[fill] [cyan](6,-1.4) node{$\gamma_{13}$};
\draw[fill] [cyan](8,1.4) node{$\gamma_{11}$};
\draw[fill] [cyan](8,-1.4) node{$\gamma_{12}$};
\end{tikzpicture}}
\centerline{\noindent {\small \textbf{Figure 3.} (Color online) The jump contour $\Sigma^{(3)}$.}}

Following the ideas in \cite{Wang-CMP}, we can prove the above proposition similarly. Then, we perform  a transformation
\begin{align}\label{49}
M^{(3)}=M^{(2)}R^{(2)},
\end{align}
where
\begin{align}
R^{(2)}=\begin{cases}
\left(\begin{array}{cc}
    1  &  R_{k1}e^{2itf}\\
    0  &  1\\
\end{array}\right),\qquad z\in\Omega_{k1},\quad
\left(\begin{array}{cc}
    1  &  0\\
    R_{k2}e^{-2itf}  &  1\\
\end{array}\right),\qquad z\in\Omega_{k2},\\
\left(\begin{array}{cc}
    1  &  -R_{k3}e^{2itf}\\
    0  &  1\\
\end{array}\right),\qquad z\in\Omega_{k3},\quad
\left(\begin{array}{cc}
    1  &  0\\
    -R_{k4}e^{-2itf}  &  1\\
\end{array}\right),\qquad z\in\Omega_{k4},\\
\left(\begin{array}{cc}
    1  &  R_{01}e^{2itf}\\
    0  &  1\\
\end{array}\right),\qquad z\in\Omega_{01}\cup\Omega_{04},\quad
\left(\begin{array}{cc}
    1  &  0\\
    R_{02}e^{-2itf}  &  1\\
\end{array}\right),\qquad z\in\Omega_{02}\cup\Omega_{03},\\
\left(\begin{array}{cc}
    1  &  0\\
    0  &  1\\
\end{array}\right),\qquad z\in\Omega_{5}\cup\Omega_{6},
\end{cases}\nonumber
\end{align}
of which $k=1,\cdots,4$.

Then, in terms of the RH problem \ref{rhp3} and Proposition \ref{prop1}, a $\bar{\partial}$-RH problem can be constructed as follows.

\begin{rhp}\label{rhp4} For $k=1,\cdots,4$,
find an analysis function $M^{(3)}(x, t, z)$ with the
following properties:\\

$\bullet$ $M^{(3)}(x, t, z)$ is analytical in $\mathbb{C}\setminus \Sigma^{(3)}$;

$\bullet$ $\bar{\partial}M^{(3)}=M^{(3)}\bar{\partial}R^{(2)}(z)$, as $\lambda\in\mathbb{C}\setminus \Sigma^{(3)}$, where
\begin{align}
\bar{\partial}R^{(2)}=\begin{cases}
\left(\begin{array}{cc}
    0  &  \bar{\partial}R_{k1}e^{2itf}\\
    0  &  0\\
\end{array}\right),\qquad z\in\Omega_{k1},
\left(\begin{array}{cc}
    0  &  0\\
    \bar{\partial}R_{k2}e^{-2itf}  &  0\\
\end{array}\right),\qquad z\in\Omega_{k2},\\
\left(\begin{array}{cc}
    0  &  -\bar{\partial}R_{k3}e^{2itf}\\
    0  &  0\\
\end{array}\right),\qquad z\in\Omega_{k3},
\left(\begin{array}{cc}
    0  &  0\\
    -\bar{\partial}R_{k4}e^{-2itf}  &  0\\
\end{array}\right),\qquad z\in\Omega_{k4},\\
\left(\begin{array}{cc}
    0  &  \bar{\partial}R_{01}e^{2itf}\\
    0  &  0\\
\end{array}\right),\qquad z\in\Omega_{01}\cup\Omega_{04},\\
\left(\begin{array}{cc}
    0  &  0\\
    \bar{\partial}R_{02}e^{-2itf}  &  0\\
\end{array}\right),\qquad z\in\Omega_{02}\cup\Omega_{03},
\left(\begin{array}{cc}
    0  &  0\\
    0  &  0\\
\end{array}\right),\qquad z\in\Omega_{5}\cup\Omega_{6};
\end{cases}\nonumber
\end{align}

$\bullet$ $M_{+}^{(3)}(x, t, z)=M^{(3)}_{-}(x, t, z)J^{(3)}(x, t, z), z\in \Sigma^{(3)}$, where
\begin{align}\label{51}
J^{(3)}(x, t, z)=\begin{cases}
\left(\begin{array}{cc}
    1  &  -R_{k1}e^{2itf}\\
    0  &  1\\
\end{array}\right),\ z\in\gamma_{k1},\qquad
\left(\begin{array}{cc}
    1  &  0\\
    R_{k2}e^{-2itf}  &  1\\
\end{array}\right),\ z\in\gamma_{k2},\\
\left(\begin{array}{cc}
    1  &  -R_{k3}e^{2itf}\\
    0  &  1\\
\end{array}\right),\ z\in\gamma_{k3},\qquad
\left(\begin{array}{cc}
    1  &  0\\
    R_{k4}e^{-2itf}  &  1\\
\end{array}\right),\ z\in\gamma_{k4},\\
\left(\begin{array}{cc}
    1  &  0\\
    (R_{44}-R_{34})e^{-2ift}  &  1\\
\end{array}\right),\ z\in \gamma_{4+},\
\left(\begin{array}{cc}
    1  &  (R_{43}-R_{33})e^{2ift}\\
    0  &  1\\
\end{array}\right),\ z\in \gamma_{4-},\\
\left(\begin{array}{cc}
    1  &  (R_{01}-R_{31})e^{2ift}\\
    0  &  1\\
\end{array}\right),\ z\in \gamma_{3+},\
\left(\begin{array}{cc}
    1  &  0\\
    (R_{02}-R_{32})e^{-2ift}  &  1\\
\end{array}\right),\ z\in \gamma_{3-},\\
\left(\begin{array}{cc}
    1  &  (R_{21}-R_{04})e^{2ift}\\
    0  &  1\\
\end{array}\right),\ z\in \gamma_{2+},\
\left(\begin{array}{cc}
    1  &  0\\
    (R_{22}-R_{03})e^{-2ift}  &  1\\
\end{array}\right),\ z\in \gamma_{2-},\\
\left(\begin{array}{cc}
    1  &  0\\
    (R_{24}-R_{14})e^{-2ift}  &  1\\
\end{array}\right),\ z\in \gamma_{1+},\
\left(\begin{array}{cc}
    1  &  (R_{23}-R_{13})e^{2ift}\\
    0  &  1\\
\end{array}\right),\ z\in \gamma_{1-};
\end{cases}
\end{align}

$\bullet$ $M^{(3)}(x, t, z)=\mathbb{I}+O(\frac{1}{z})$, as $z\rightarrow\infty$;

$\bullet$ $M^{(3)}(x, t, z)=\frac{\sigma_{2}}{z}+O(1)$ as $z\rightarrow 0$.
\end{rhp}

\subsection{The decomposition of the mixed $\bar{\partial}$-RH problem }
In this part, we devote to decompose the mixed $\bar{\partial}$-RH problem, i.e., RH problem \ref{rhp4}, into two parts, including a pure RH problem with $\bar{\partial}R^{(2)}=0$ and a pure $\bar{\partial}$-RH problem with $\bar{\partial}R^{(2)}\neq0$. Firstly, to solve the pure RH problem with $\bar{\partial}R^{(2)}=0$, we construct a model $M^{RHP}$.

\subsubsection{Pure RH problem}
For the model RH problem $M^{RHP}$, the following conclusion is established.

\begin{rhp}\label{rhp5}
Find an analysis function $M^{RHP}(x, t, z)$ with the
following properties:\\

$\bullet$ $M^{RHP}(x, t, z)$ is meromorphic in $\mathbb{C}\setminus \Sigma^{(3)}$;\\

$\bullet$ $M_{+}^{RHP}(x, t, z)=M^{RHP}_{-}(x, t, z)J^{(3)}(x, t, z), z\in \Sigma^{(2)}$, where $J^{(3)}(x, t, z)$ is given in \eqref{51}.

$\bullet$ $M^{RHP}(x, t, z)=\mathbb{I}+O(\frac{1}{z})$, as $z\rightarrow\infty$.

$\bullet$ $M^{RHP}(x, t, z)=\frac{\sigma_{2}}{z}+O(1)$ as $z\rightarrow 0$;
\end{rhp}

Then, in order to construct the solution $M^{RHP}(x, t, z)$  for the RH problem \ref{rhp5}, we decompose $M^{RHP}(x, t, z)$ into following form
\begin{align}\label{54}
M^{RHP}(z)=\begin{cases}
E(z)M^{out}(z),\qquad\qquad  z\in \mathbb{C}\backslash \mathcal{U}_{\xi},\\
E(z)M^{out}(z)M^{(\xi_{1})},\qquad z\in \mathcal{U}_{\xi_{1}},\\
E(z)M^{out}(z)M^{(\xi_{2})},\qquad z\in \mathcal{U}_{\xi_{2}},\\
E(z)M^{out}(z)M^{(\xi_{3})},\qquad z\in \mathcal{U}_{\xi_{3}},\\
E(z)M^{out}(z)M^{(\xi_{4})},\qquad z\in \mathcal{U}_{\xi_{4}},\\
\end{cases}
 \end{align}
where
\begin{align}
\mathcal{U}_{\xi}=\mathcal{U}_{\xi_{1}}\cup\mathcal{U}_{\xi_{2}}\cup\mathcal{U}_{\xi_{3}}\cup\mathcal{U}_{\xi_{4}},\ \mathcal{U}_{\xi_{k}}=\{z:\left\vert z-\xi_{k}\right\vert\leq \rho\},\ k=1,\cdots,4,\nonumber
\end{align}
and $M^{out}(z)$ is a solution by ignoring the jump conditions of $M^{RHP}(z)$. $M^{(\xi_{k})}$ can be reduced to the  parabolic cylinder model
and $E(z)$ is the solution of a small-norm RH problem.
\begin{prop}\label{prop2} For $1\leq p\leq +\infty$, as $t\rightarrow +\infty$,
the jump matrix $J^{(3)}$ given in \eqref{51} meets the following estimates
\begin{align}
\left\Vert J^{(3)}(z)-\mathbb{I}\right\Vert_{L^{p}(\Sigma^{(3)}\setminus\mathcal{U}_{\xi} )}=O(e^{-h_{p}t}),\nonumber
\end{align}
where $h_{p}>0$ is a constant.
\end{prop}
\begin{proof}
We prove the case of $z\in \gamma_{24}\setminus\mathcal{U}_{\xi_{2}}$ and $z\in\gamma_{1+}$, other cases can be shown in a similar way.

For $z\in \gamma_{24}\setminus\mathcal{U}_{\xi_{2}}$, when $1\leq p\leq +\infty$, by using \eqref{48} and \eqref{51}, we have
\begin{align}
\left\Vert J^{(3)}(z)-\mathbb{I}\right\Vert_{L^{p}(\gamma_{24}\setminus\mathcal{U}_{\xi_{2}} )}=\left\Vert R_{24}e^{-2itf}\right\Vert_{L^{p}(\gamma_{24}\setminus\mathcal{U}_{\xi_{2}} )}\lesssim \left\Vert e^{-2itf}\right\Vert_{L^{p}(\gamma_{24}\setminus\mathcal{U}_{\xi_{2}} )}.\nonumber
 \end{align}
Denote $z=\xi_{2}+e^{i\phi}d_{2}=\xi_{2}+u+iv$, we have $Re(2if)\geq (1+\vert z \vert^{-2})v^{2}$, which leads to
\begin{align}
\left\Vert e^{-2itf}\right\Vert_{L^{p}(\gamma_{24}\setminus\mathcal{U}_{\xi_{2}} )}\lesssim t^{-1}e^{-h_{p}t}.\nonumber
 \end{align}

For $z\in \gamma_{1+}$,
\begin{align}
\left\Vert J^{(3)}(z)-\mathbb{I}\right\Vert_{L^{p}(\gamma_{1+})}=\left\Vert (R_{24}-R_{14})e^{-2itf}\right\Vert_{L^{p}(\gamma_{1+})}\lesssim \left\Vert e^{-2itf}\right\Vert_{L^{p}(\gamma_{1+})}\lesssim t^{-1/p}e^{-h_{p}t}.\nonumber
 \end{align}
\end{proof}
From the Proposition \ref{prop2}, we can find that the jump $J^{(3)}$ uniformly goes to $\mathbb{I}$ outside $\mathcal{U}_{\xi}$. Moreover, the outside model $M^{out}(z)$ arrives at the following RH problem

\begin{rhp}\label{rhp5.1}
 Find an analysis function $M^{out}(x, t, z)$ with the
following properties:\\

$\bullet$ $M^{out}(x, t, z)$ is analytic in $\mathbb{C}\setminus \{0\}$;\\

$\bullet$ $M^{out}(x, t, z)=\mathbb{I}+O(\frac{1}{z})$, as $z\rightarrow\infty$;

$\bullet$ $M^{out}(x, t, z)=\frac{\sigma_{2}}{z}+O(1)$ as $z\rightarrow 0$.
\end{rhp}

\begin{prop}\label{prop2.1} The RH problem \ref{rhp5.1} can be uniquely solved as
\begin{align}
M^{out}(x, t, z)=\mathbb{I}+\frac{\sigma_{2}}{z}.\nonumber
 \end{align}
\end{prop}

Since there is not a uniform estimate for $J^{(3)}(z)-\mathbb{I}$ on $z\in \mathcal{U}_{\xi}$ as $t\rightarrow\infty$,
we need to introduce a local solvable model $M^{(\xi_{k})}$ to match the jumps of $M^{RHP}$ on
$\Sigma^{(3)}\cap \mathcal{U}_{\xi}$. To solve the model $M^{(\xi_{k})}$, we shall carry out the scaling transformation to separate the time $t$
from the jump matrix by introducing scaling the transformation
\begin{align}\label{60}
s=s(z)=\sqrt{2t\epsilon_{k}f''(\xi_{k})}(z-\xi_{k}),\quad k=1,\cdots, 4,
 \end{align}
where $f(z)=f(\xi_{k})+\frac{f''(\xi_{k})}{2}(z-\xi_{k})^{2}+O(\vert z-\xi_{k}\vert^{3}), z\rightarrow\xi_{k}$, and introduce the scaling operators
\begin{align}
&g(z)\mapsto (N_{k}g)(z)=g(\frac{s}{\sqrt{2t\epsilon_{k}f''(\xi_{k})}}+\xi_{k}).\nonumber
 \end{align}
Therefore, we have
\begin{align}
&N_{k}(R_{k1}e^{2itf})=\frac{\overline{r_{\xi_{k}}}}{1-\vert r_{\xi_{k}}\vert^{2}}s^{-2i\nu(\xi_{k})\epsilon_{k}}e^{\frac{is^{2}}{2\epsilon_{k}}},\notag\\
&N_{k}(R_{k2}e^{-2itf})=\frac{r_{\xi_{k}}}{1-\vert r_{\xi_{k}}\vert^{2}}s^{2i\nu(\xi_{k})\epsilon_{k}}e^{-\frac{is^{2}}{2\epsilon_{k}}},\notag\\
&N_{k}(R_{k3}e^{2itf})=\overline{r_{\xi_{k}}}s^{-2i\nu(\xi_{k})\epsilon_{k}}e^{\frac{is^{2}}{2\epsilon_{k}}},\notag\\
&N_{k}(R_{k4}e^{-2itf})=r_{\xi_{k}}s^{2i\nu(\xi_{k})\epsilon_{k}}e^{-\frac{is^{2}}{2\epsilon_{k}}},\nonumber
\end{align}
where we have taken $r_{\xi_{k}}=r(\xi_{k})\delta_{0}^{2}(\xi_{k})e^{-2itf(\xi_{k})-i\epsilon_{k}\nu(\xi_{k})\ln(2t\epsilon_{k}f''(\xi_{k}))}$.

After that, we can generate the following RH problem $M^{(\xi_{k})}$ as $t\rightarrow\infty$ in
the $s$ plane, see Fig. 4.\\
\centerline{\begin{tikzpicture}
\draw[fill] (0,0) circle [radius=0.035];
\draw[->][thick](-2.12,-2.12)--(-1.06,-1.06);
\draw[-][thick](-1.06,-1.06)--(0,0)node[below]{$O$};
\draw[->][thick](0,0)--(1.06,1.06);
\draw[-][thick](1.06,1.06)--(2.12,2.12);
\draw[->][thick](-2.12,2.12)--(-1.06,1.06);
\draw[-][thick](-1.06,1.06)--(0,0);
\draw[->][thick](0,0)--(1.06,-1.06);
\draw[->][dashed](-3,0)--(-1.5,0);
\draw[->][dashed](-1.5,0)--(1.5,0);
\draw[-][dashed](1.5,0)--(3.0,0)node[right]{$\mathbb{R}$};
\draw[-][thick](1.06,-1.06)--(2.12,-2.12);
\draw[fill] [red] (0,1) node{$\Omega_{0}$};
\draw[fill] [red](0,-1) node{$\Omega_{0}$};
\draw[fill] [red](-1,0.5) node{$\Omega_{4}$};
\draw[fill] [red](1,0.5) node{$\Omega_{1}$};
\draw[fill] [red](-1,-0.5) node{$\Omega_{3}$};
\draw[fill] [red](1,-0.5) node{$\Omega_{2}$};
\draw[fill] [cyan] (-1.6,2.0) node{$\gamma_{4}$};
\draw[fill] [cyan] (-1.6,-2.0) node{$\gamma_{3}$};
\draw[fill] [cyan] (1.6,2.0) node{$\gamma_{1}$};
\draw[fill] [cyan] (1.6,-2.0) node{$\gamma_{2}$};
\end{tikzpicture}}
\centerline{\noindent {\small \textbf{Figure 4.} (Color online) The jump contour $\widetilde{\Sigma}=\gamma_{1}\cup\gamma_{2}\cup\gamma_{3}\cup\gamma_{4}$  and domains $\Omega_{j}(j=0,\cdots, 4)$.}}

\begin{rhp}\label{rhp6}
The analysis function $M^{(\xi_{k})}(s)$ admits the
following properties:\\

 $\bullet$ $M^{(\xi_{k})}(s)$ is analytic in $\mathbb{C}\setminus \widetilde{\Sigma}$;\\

 $\bullet$ $M_{+}^{(\xi_{k})}(s)=M^{(\xi_{k})}_{-}(s)J^{(\xi_{k})}(s), z\in \widetilde{\Sigma}$, where
\begin{align}
&k=1, 3:\notag\\
& J^{(\xi_{k})}(s)=
\begin{cases}
\left(\begin{array}{cc}
    1  &  -\frac{\overline{r_{\xi_{k}}}}{1-\vert r_{\xi_{k}}\vert^{2}}s^{-2i\nu(\xi_{k})}e^{\frac{is^{2}}{2}}\\
    0  &  1\\
\end{array}\right),\qquad s\in\gamma_{1},\\
\left(\begin{array}{cc}
    1  &  0\\
    \frac{r_{\xi_{k}}}{1-\vert r_{\xi_{k}}\vert^{2}}s^{2i\nu(\xi_{k})}e^{-\frac{is^{2}}{2}}  &  1\\
\end{array}\right),\qquad s\in\gamma_{2},\\
\left(\begin{array}{cc}
    1  &  -\overline{r_{\xi_{k}}}s^{-2i\nu(\xi_{k})}e^{\frac{is^{2}}{2}}\\
    0  &  1\\
\end{array}\right),\qquad s\in\gamma_{3},\\
\left(\begin{array}{cc}
    1  &  0\\
    r_{\xi_{k}}s^{2i\nu(\xi_{k})}e^{-\frac{is^{2}}{2}}  &  1\\
\end{array}\right),\qquad s\in\gamma_{4},
\end{cases}\notag\\
&k=2, 4:\notag\\
& J^{(\xi_{k})}(s)=
\begin{cases}
\left(\begin{array}{cc}
    1  &  0\\
    r_{\xi_{k}}s^{-2i\nu(\xi_{k})}e^{\frac{is^{2}}{2}}  &  1\\
\end{array}\right),\qquad s\in\gamma_{1},\\
\left(\begin{array}{cc}
    1  &  -\overline{r_{\xi_{k}}}s^{2i\nu(\xi_{k})}e^{-\frac{is^{2}}{2}}\\
    0  &  1\\
\end{array}\right),\qquad s\in\gamma_{2},\\
\left(\begin{array}{cc}
    1  &  0\\
    \frac{r_{\xi_{k}}}{1-\vert r_{\xi_{k}}\vert^{2}}s^{-2i\nu(\xi_{k})}e^{\frac{is^{2}}{2}}  &  1\\
\end{array}\right),\qquad s\in\gamma_{3},\\
\left(\begin{array}{cc}
    1  &  -\frac{\overline{r_{\xi_{k}}}}{1-\vert r_{\xi_{k}}\vert^{2}}s^{2i\nu(\xi_{k})}e^{-\frac{is^{2}}{2}}\\
    0  &  1\\
\end{array}\right),\qquad s\in\gamma_{4};
\end{cases}\nonumber
\end{align}

$\bullet$ $M^{(\xi_{k})}(s)=\mathbb{I}+O(\frac{1}{s})$ as $s\rightarrow\infty$.
\end{rhp}
It is well known that the  solution $M^{(\xi_{k})}(s)$ of  RH problem \ref{rhp6} can be solved explicitly via using the parabolic
cylinder  model as shown in ``Appendix \ref{A}".

Due to the jump matrix of  $M^{(\xi_{k})}$ and $M^{RHP}$ is coincident in disk $\mathcal{U}_{\xi}$, the matrix $E(z)$ erects the jump of $M^{RHP}$ inside disk $\mathcal{U}_{\xi}$, and there is still a jump from $M^{RHP}$ outside the disk, so the jump path of $E(z)$ is
\begin{align}
\Sigma^{E}=\partial\mathcal{U}_{\xi}\cup
\left(\Sigma^{(3)}\setminus\mathcal{U}_{\xi}\right),\nonumber
\end{align}
with clockwise direction for $\partial\mathcal{U}_{\xi}$. Then, the $E(z)$ satisfies the following
RH problem, see Fig. 5.\\
\centerline{\begin{tikzpicture}[scale=0.8]
\draw[<-][thick](-0.5,0) arc(0:360:0.5);
\draw[-][thick](-0.5,0) arc(0:30:0.5);
\draw[-][thick](-0.5,0) arc(0:150:0.5);
\draw[-][thick](-0.5,0) arc(0:210:0.5);
\draw[-][thick](-0.5,0) arc(0:330:0.5);
\draw[<-][thick](3.5,0) arc(0:360:0.5);
\draw[-][thick](3.5,0) arc(0:30:0.5);
\draw[-][thick](3.5,0) arc(0:150:0.5);
\draw[-][thick](3.5,0) arc(0:210:0.5);
\draw[-][thick](3.5,0) arc(0:330:0.5);
\draw[<-][thick](8,0) arc(0:360:1);
\draw[-][thick](8,0) arc(0:30:1);
\draw[-][thick](8,0) arc(0:150:1);
\draw[-][thick](8,0) arc(0:210:1);
\draw[-][thick](8,0) arc(0:330:1);
\draw[<-][thick](-4,0) arc(0:360:1);
\draw[-][thick](-4,0) arc(0:30:1);
\draw[-][thick](-4,0) arc(0:150:1);
\draw[-][thick](-4,0) arc(0:210:1);
\draw[-][thick](-4,0) arc(0:330:1);
\draw[-][thick](5,-2)--(5,-1);
\draw[->][thick](5,0)--(5,-1);
\draw[->][thick](5,2)--(5,1);
\draw[-][thick](5,0)--(5,1);
\draw[-][thick](2,-1)--(2,-0.5);
\draw[->][thick](2,0)--(2,-0.5);
\draw[->][thick](2,1)--(2,0.5);
\draw[-][thick](2,0)--(2,0.5);
\draw[-][thick](0,-1)--(0,-0.5);
\draw[->][thick](0,0)--(0,-0.5);
\draw[->][thick](0,1)--(0,0.5);
\draw[-][thick](0,0)--(0,0.5);
\draw[-][thick](-3,-2)--(-3,-1);
\draw[->][thick](-3,0)--(-3,-1);
\draw[->][thick](-3,2)--(-3,1);
\draw[-][thick](-3,0)--(-3,1);
\draw[<-][thick](-2,-1)--(-3,-2);
\draw[-][thick](-3,-2)--(-4,-1);
\draw[<-][thick](-4,-1)--(-4.3,-0.7);
\draw[-][thick](-5.7,-0.7)--(-6,-1);
\draw[<-][thick](-6,-1)--(-7,-2);
\draw[<-][thick](-2,1)--(-3,2);
\draw[-][thick](-3,2)--(-4,1);
\draw[<-][thick](-4,1)--(-4.3,0.7);
\draw[-][thick](-5.7,0.7)--(-6,1);
\draw[<-][thick](-6,1)--(-7,2);
\draw[->][thick](0,1)--(0.5,0.5);
\draw[-][thick](0.5,0.5)--(1,0);
\draw[->][thick](0,-1)--(0.5,-0.5);
\draw[-][thick](0.5,-0.5)--(1,0);
\draw[-][thick](2,-1)--(1.5,-0.5);
\draw[<-][thick](1.5,-0.5)--(1,0);
\draw[-][thick](2,1)--(1.5,0.5);
\draw[<-][thick](1.5,0.5)--(1,0);
\draw[-][thick](0,1)--(-0.5,0.5);
\draw[->][thick](-0.65,0.35)--(-0.5,0.5);
\draw[-][thick](0,-1)--(-0.5,-0.5);
\draw[->][thick](-0.65,-0.35)--(-0.5,-0.5);
\draw[-][thick](-1.35,0.35)--(-1.5,0.5);
\draw[-][thick](-1.5,0.5)--(-2,1);
\draw[-][thick](-1.35,-0.35)--(-1.5,-0.5);
\draw[-][thick](-1.5,-0.5)--(-2,-1);
\draw[fill] (1,0) circle [radius=0.035] node[below]{$0$};
\draw[-][thick](4,1)--(5,2);
\draw[->][thick](5,2)--(6,1);
\draw[-][thick](6,1)--(6.3,0.7);
\draw[->][thick](7.7,0.7)--(8,1);
\draw[-][thick](8,1)--(9,2);
\draw[-][thick](4,-1)--(5,-2);
\draw[->][thick](5,-2)--(6,-1);
\draw[-][thick](6,-1)--(6.3,-0.7);
\draw[->][thick](7.7,-0.7)--(8,-1);
\draw[-][thick](8,-1)--(9,-2);
\draw[->][thick](2,1)--(2.5,0.5);
\draw[-][thick](2.5,0.5)--(2.65,0.35);
\draw[->][thick](2,-1)--(2.5,-0.5);
\draw[-][thick](2.5,-0.5)--(2.65,-0.35);
\draw[-][thick](3.35,0.35)--(3.5,0.5);
\draw[->][thick](3.5,0.5)--(4,1);
\draw[-][thick](3.35,-0.35)--(3.5,-0.5);
\draw[->][thick](3.5,-0.5)--(4,-1);
\draw[fill] [cyan] (-5,-1.4) node{$\partial\mathcal{U}_{\xi_{4}}$};
\draw[fill] [cyan] (-1,-0.9) node{$\partial\mathcal{U}_{\xi_{3}}$};
\draw[fill] [cyan] (3,-0.9) node{$\partial\mathcal{U}_{\xi_{2}}$};
\draw[fill] [cyan] (7,-1.4) node{$\partial\mathcal{U}_{\xi_{1}}$};
\end{tikzpicture}}
\centerline{\noindent {\small \textbf{Figure 5.} (Color online) The jump contour $\Sigma^{E}$.}}

\begin{rhp}\label{rhp7}
Find a matrix-valued function $E(z)$ has the
following properties:\\

 $\bullet$ $E(z)$ is analytic in $\mathbb{C}\setminus \Sigma^{E}$;\\

 $\bullet$ $E_{+}(z)=E_{-}(z)J^{E}(z), z\in \Sigma^{E}$, where
\begin{align}\label{65}
J^{E}(z)=
\begin{cases}
M^{out}(z)J^{(3)}(M^{out})^{-1}(z),\qquad z\in\Sigma^{(3)}\setminus\mathcal{U}_{\xi},\\
M^{out}(z)M^{(\xi_{k})}(M^{out})^{-1}(z),\qquad z\in \partial\mathcal{U}_{\xi_{k}};
\end{cases}
\end{align}

$\bullet$ $E(z)=\mathbb{I}+O(\frac{1}{z})$ as $z\rightarrow\infty$.
\end{rhp}

\begin{prop}\label{prop3}
The jump matrix $J^{E}$ shown  in \eqref{65} has the following estimates
\begin{align}\label{66}
\left\vert J^{E}(z)-\mathbb{I}\right\vert=\begin{cases}
O(e^{-h_{p}t}),\qquad z\in \Sigma^{(3)}\setminus\mathcal{U}_{\xi},\\
O(t^{-1/2}),\qquad z\in \partial\mathcal{U}_{\xi}.
\end{cases}
\end{align}
\end{prop}
\begin{proof}
For $z\in\Sigma^{(3)}\setminus\mathcal{U}_{\xi}$, using the definition of $J^{E}$ \eqref{65}, one has
\begin{align}
\left\vert J^{E}(z)-\mathbb{I}\right\vert\lesssim\left\vert J^{(3)}-\mathbb{I}\right\vert.\nonumber
 \end{align}
Thus, the estimates for $\left\vert J^{E}(z)-\mathbb{I}\right\vert$ is consistent with Proposition \ref{prop2}.

For $z\in \partial\mathcal{U}_{\xi}$, the variable $s=\sqrt{2t\epsilon_{k}f''(\xi_{k})}(z-\xi_{k})$ tends to infinity as $t\rightarrow \infty$, and according to the asymptotic expansion
 \begin{align}
M^{(\xi_{k})}(s)=\mathbb{I}+\frac{M^{(\xi_{k})}_{1}}{s}+O(\frac{1}{s^{2}}),\qquad s\rightarrow\infty,\nonumber
 \end{align}
it is easy to arrive at
 \begin{align}
M^{(\xi_{k})}(s)=\mathbb{I}+\frac{M^{(\xi_{k})}_{1}}{\sqrt{2t\epsilon_{k}f''(\xi_{k})}(z-\xi_{k})}+O(\frac{1}{t}),\qquad t\rightarrow\infty,\nonumber
\end{align}
which results into $\left\vert M^{(\xi_{k})}(s)-\mathbb{I} \right\vert=O(t^{-\frac{1}{2}})$. Applying the definition of $J^{E}$ \eqref{65}, we obtain
\begin{align}
\left\vert J^{E}(z)-\mathbb{I}\right\vert\lesssim\left\vert M^{(\xi_{k})}-\mathbb{I}\right\vert=O(t^{-\frac{1}{2}}).\nonumber
 \end{align}
\end{proof}
The fact that estimates $\left\vert J^{E}(z)-\mathbb{I}\right\vert$ in \eqref{66} decay uniformly illustrates
the RH problem \ref{rhp7} is a small-norm RH problem,  and its existence
and uniqueness have been verified in \cite{Li-WKI36,Li-WKI37}.

According to Beals-Coifman theorem, in order to construct the solution of RH problem \ref{rhp7}, we
decompose  the jump matrix $J^{E}$ into
\begin{align}
J^{E}=(b_{-})^{-1}b_{+},\quad b_{-}=\mathbb{I},\quad b_{+}=J^{E},\nonumber
 \end{align}
and
\begin{align}\label{72}
&(\omega_{E})_{-}=\mathbb{I}-b_{-}=0,\quad (\omega_{E})_{+}=b_{+}-\mathbb{I}=J^{E}-\mathbb{I},\notag\\
&\omega_{E}=(\omega_{E})_{+}+(\omega_{E})_{-}=J^{E}-\mathbb{I},\notag\\
&C_{\omega_{E}}\mathcal{F}=C_{-}(\mathcal{F}(\omega_{E})_{+})+C_{+}(\mathcal{F}(\omega_{E})_{-})=C_{-}(\mathcal{F}(J^{E}-\mathbb{I})),
 \end{align}
where $C_{-}$ is the Cauchy operator
\begin{align}\label{73}
C_{-}\mathcal{F}(z)=\lim_{z'\rightarrow z\in\Sigma^{E}}\frac{1}{2\pi i}\int_{\Sigma^{E}}\frac{\mathcal{F}(\zeta)}{\zeta-z'}\mathrm{d}\zeta.
 \end{align}
Finally, the solution for RH problem \ref{rhp7} admits
\begin{align}
E(z)=\mathbb{I}+\frac{1}{2\pi i}\int_{\Sigma^{E}}\frac{\mu_{E}(\zeta)(J^{E}-\mathbb{I})}{\zeta-z}\mathrm{d}\zeta,\nonumber
 \end{align}
where $\mu_{E}\in L^{2}(\Sigma^{E})$  solves $(1-C_{\omega_{E}})\mu_{E}=\mathbb{I}$.

Proposition \ref{prop3} declares
\begin{align}\label{75}
\left\Vert J^{E}-\mathbb{I}\right\Vert_{L^{p}(\Sigma^{E})}=O(t^{-\frac{1}{2}}), \quad p\in[1,+\infty).
 \end{align}
Moreover, in terms of the properties of the
Cauchy operator $C_{-}$ and \eqref{75}, we get
\begin{align}\label{76}
&\left\Vert C_{\omega_{E}}\right\Vert_{L^{2}(\Sigma^{E})}\lesssim \left\Vert C_{-}\right\Vert_{L^{2}(\Sigma^{E})}\left\Vert J^{E}-\mathbb{I}\right\Vert_{L^{\infty}(\Sigma^{E})}\lesssim O(t^{-\frac{1}{2}}),\notag\\
&\left\Vert \mu_{E}-\mathbb{I}\right\Vert_{L^{2}(\Sigma^{E})}=\left\Vert C_{\omega_{E}}\mu_{E}\right\Vert_{L^{2}(\Sigma^{E})}=\left\Vert C_{-}\right\Vert_{L^{2}(\Sigma^{E})} \left\Vert\mu_{E}\right\Vert_{L^{2}(\Sigma^{E})}\left\Vert J^{E}-\mathbb{I}\right\Vert_{L^{\infty}(\Sigma^{E})}\lesssim O(t^{-\frac{1}{2}}).
 \end{align}
The first formula of \eqref{76} implies that $(1-C_{\omega_{E}})^{-1}$ is existent. As a result, the existence and
uniqueness of $\mu_{E}$ and $E(z)$ are confirmed. Thus, it is reasonable
to construct $M^{RHP}$ in \eqref{54}.

Furthermore, to reconstruct the solutions of $q(x, t)$, we need to
discuss the asymptotic behavior of $E(z)$ as $z\rightarrow \infty$ and asymptotic
behavior of $E_{1}$ as $t\rightarrow\infty$. Finally, considering the estimates \eqref{75} and \eqref{76}, we derive the following asymptotic expansion as $z\rightarrow \infty$, given by
\begin{align}
E(z)=\mathbb{I}+\frac{E_{1}}{z}+O(z^{-2}), \qquad z\rightarrow \infty,\nonumber
 \end{align}
where
\begin{align}
E_{1}=-\frac{1}{2\pi i}\int_{\Sigma^{E}}\mu_{E}(\zeta)(J^{E}-\mathbb{I})\mathrm{d}\zeta.\nonumber
 \end{align}
Then, as $t\rightarrow\infty$, the asymptotic behavior of $E_{1}$ is
\begin{align}\label{79}
E_{1}&=-\sum_{k=1}^{4}\frac{1}{2\pi i}\oint_{\partial\mathcal{U}_{\xi_{k}}}(J^{E}-\mathbb{I})\mathrm{d}\zeta+O(t^{-1})\notag\\
= &-\sum_{k=1}^{4}\frac{1}{2\pi i}\oint_{\partial\mathcal{U}_{\xi_{k}}}\frac{M^{(\xi_{k})}_{1}}{\sqrt{2t\epsilon_{k}f''(\xi_{k})}(\zeta-\xi_{k})}\mathrm{d}\zeta+O(t^{-1})\notag\\
=&-\sum_{k=1}^{4}\frac{M^{(\xi_{k})}_{1}}{\sqrt{2t\epsilon_{k}f''(\xi_{k})}}+O(t^{-1}).
\end{align}

\subsubsection{Pure $\bar{\partial}$-RH problem }
In this subsection, we mainly discuss the pure $\bar{\partial}$-problem at the case of $\bar{\partial}R^{(2)}=0$. Define
\begin{align}\label{81}
M^{(4)}(z)=M^{(3)}(z)\left(M^{RHP}(z)\right)^{-1},
\end{align}
which  is continuous and has no jumps in the complex plane. The pure $\bar{\partial}$-problem $M^{(4)}(z)$ is given as follows.

\begin{rhp}\label{rhp8}
Find a matrix-valued function $M^{(4)}(z)$ has the
following properties:\\

$\bullet$ $M^{(4)}(z)$  is continuous in  $\mathbb{C}\setminus (\mathbb{R}\cup\Sigma^{(3)})$;\\

$\bullet$ $\bar{\partial}M^{(4)}(z)=M^{(4)}(z)W^{(4)}(z), z\in \mathbb{C}$, where
\begin{align}
W^{(4)}(z)=M^{RHP}(z)\bar{\partial}R^{(2)}(z)\left(M^{RHP}(z)\right)^{-1};\nonumber
\end{align}

$\bullet$ $M^{(4)}(z)=\mathbb{I}+O(\frac{1}{z})$ as $z\rightarrow\infty$.
\end{rhp}

The pure $\bar{\partial}$-problem can be solved as  the following integral equation
\begin{align}\label{83}
M^{(4)}(z)=\mathbb{I}-\frac{1}{\pi}\iint_{\mathbb{C}}\frac{M^{(4)}(\zeta)W^{(4)}(\zeta)}{\zeta-z}\mathrm{d}A(\zeta),
\end{align}
where $\mathrm{d}A(\zeta)$ is the Lebesgue measure. Further, the equation \eqref{83} can be written into operator form
\begin{align}
(\mathbb{I}-\mathcal{S})M^{(4)}(z)=\mathbb{I},\nonumber
\end{align}
where $\mathcal{S}$ is the Cauchy operator
\begin{align}
\mathcal{S}[\mathcal{F}](z)=-\frac{1}{\pi}\iint_{\mathbb{C}}\frac{\mathcal{F}(\zeta)W^{(4)}(\zeta)}{\zeta-z}\mathrm{d}A(\zeta).\nonumber
\end{align}

\begin{prop}\label{prop4}
For large time $t$,
\begin{align}
\left\Vert \mathcal{S}\right\Vert_{L^{\infty}\rightarrow L^{\infty}}\lesssim t^{-\frac{1}{4}},\nonumber
\end{align}
which denotes that the operator $\mathbb{I}-\mathcal{S}$ is invertible and the solution of pure $\bar{\partial}$-problem uniquely exists.
\end{prop}

Next, in order  to reconstruct the potential $q(x, t)$ as $t\rightarrow\infty$,
the long
time asymptotic behaviors of $M^{(4)}_{1}$ need be studied. They are presented in the
asymptotic expansion of $M^{(4)}(z)$ as $z\rightarrow \infty$, i.e.
\begin{align}
M^{(4)}(z)=\mathbb{I}+\frac{M^{(4)}_{1}}{z}+O(z^{-2}), \quad z\rightarrow \infty,\nonumber
 \end{align}
\begin{align}
M^{(4)}_{1}=\frac{1}{\pi}\iint_{\mathbb{C}}M^{(4)}(\zeta)W^{(4)}(\zeta)\mathrm{d}A(\zeta).\nonumber
\end{align}
The $M^{(4)}_{1}$ admits the following proposition.

\begin{prop}\label{prop5}
For large time $t$, $M^{(4)}_{1}$ meets the following inequality
\begin{align}
\left\vert M^{(4)}_{1}\right\vert \lesssim t^{-\frac{3}{4}}.\nonumber
 \end{align}
\end{prop}

\subsection{The final step}
Now, we are going to give the final proof of the theorem \ref{thm1} as $t\rightarrow \infty$. According to the transformations \eqref{49}, \eqref{54} and \eqref{81}, one has
\begin{align}
M(z)=M^{(4)}(z)E(z)M^{out}(z)(R^{(2)}(z))^{-1}\delta^{-\sigma_{3}},\quad z\in \mathbb{C}\backslash \mathcal{U}_{\xi}.\nonumber
\end{align}
Then, we take $z\rightarrow \infty$ along the imaginary
axis such that $R^{(2)}(\lambda)=\mathbb{I}$, after calculation, we get
\begin{align}
q(x,t)=-1-i(E_{1})_{12}+O(t^{-3/4}).\nonumber
\end{align}
According to \eqref{79}, \eqref{A9}  and \eqref{A10}, we finally derive
\begin{align}
q(x,t)=-1+\sum_{k=1}^{4}\frac{\sqrt{2\pi}e^{-\frac{\pi}{4}i\epsilon_{k}-\frac{\pi\nu(\xi_{k})}{2}}}{\sqrt{2t\epsilon_{k}f''
(\xi_{k})}r_{\xi_{k}}\Gamma(i\epsilon_{k}\nu(\xi_{k}))}+O(t^{-3/4}),\nonumber
\end{align}
where $\nu(\xi_{k})$ is given in \eqref{34}.

\section{Painleve asymptotics in transition region}
In this section, we aim to study the asymptotics in the region $-C<(\xi+8)t^{2/3}<0$ with $C>0$
which corresponds to Fig. 2(b). In this situation, the two stationary points $\xi_{4}$ and $\xi_{3}$ are real and approach to $z=-1$ at least the speed of $t^{-1/3}$ as $t\rightarrow +\infty$.

The first deformation for the transition region is the same as the first deformation in the oscillating region except that
the change of the fourth property in Proposition \ref{pjia}. The fourth property in Proposition \ref{pjia}, named the local property for $\delta(z)$, shall hold true for $k=1,2$ instead of $k=1,\cdots,4$. Besides, the interval $I_{1}$ is amended to $(\xi_{2},\xi_{1})$, and $I_{2}$ become $\mathbb{R}\setminus I_{1}$. Next, following the RH problem \ref{rhp3}, in order to remove the singularity $z=0$, we make the following transformation
\begin{align}\label{102}
M^{(2)}(z)=\left(I+\frac{\sigma_{2}}{z}M^{(3)}(0)^{-1}\right)M^{(3)}(z),
\end{align}
which makes $M^{(4)}(z)$ turns into the RH problem without spectral singularity.
\begin{rhp}\label{rhp10}
Find an analysis function $M^{(3)}(x, t, z)$ with the
following properties:\\

 $\bullet$ $M^{(3)}(x, t, z)$ is analytical in $\mathbb{C}\setminus \mathbb{R}$;\\

 $\bullet$ $M^{(3)}_{+}(x, t, z)=M^{(3)}_{-}(x, t, z)J^{(3)}(x, t, z), z\in \mathbb{R}$, where
$J^{(3)}(x, t, z)=J^{(2)}(x, t, z)$  is given by \eqref{43};

$\bullet$ $M^{(3)}(x, t, z)=\mathbb{I}+O(\frac{1}{z})$ as $z\rightarrow\infty$.\\

\end{rhp}

\subsection{The construction of the mixed $\bar{\partial}$-RH problem}
In this section, we open the jump contour  off the real axis by the $\bar{\partial}$ extension. Denote \begin{align}
&\gamma_{4}=\xi_{4}+e^{i(\pi-\phi)}\mathbb{R}_{+},\
\gamma_{3}=\xi_{3}+e^{i\phi}\widetilde{d_{1}}, \ \gamma_{01}=\xi_{0}+e^{i(\pi-\phi)}\widetilde{d_{1}}, \ \gamma_{04}=\xi_{0}+e^{i\phi}d_{1}, \notag\\
&\gamma_{21}=\xi_{2}+e^{i(\pi-\phi)}d_{1},\ \gamma_{24}=\xi_{2}+e^{i\phi}d_{2},\ \gamma_{14}=\xi_{1}+e^{i(\pi-\phi)}d_{2}, \gamma_{11}=\xi_{1}+e^{i\phi}\mathbb{R}_{+},\notag\\
&\gamma_{02}=\overline{\gamma_{01}}, \ \gamma_{03}=\overline{\gamma_{04}}, \ \gamma_{22}=\overline{\gamma_{21}}, \ \gamma_{23}=\overline{\gamma_{24}}, \ \gamma_{13}=\overline{\gamma_{14}}, \ \gamma_{12}=\overline{\gamma_{11}},\notag\\
&\gamma_{3+}=\frac{\xi_{3}}{2}+e^{\frac{i\pi}{2}}\widetilde{d_{4}},\ \gamma_{2+}=\frac{\xi_{2}}{2}+e^{\frac{i\pi}{2}}d_{4},\ \gamma_{1+}=\frac{\xi_{2}+\xi_{1}}{2}+e^{\frac{i\pi}{2}}d_{3}, \notag\\
&\gamma_{3-}=\overline{\gamma_{3+}},\ \gamma_{2-}=\overline{\gamma_{2+}},\ \gamma_{1-}=\overline{\gamma_{1+}},\nonumber
\end{align}
where $\xi_{0}=0$ and
\begin{align}
&d_{1}\in (0,\frac{\xi_{2}}{2\cos \phi}),\ d_{2}\in (0,\frac{\xi_{1}-\xi_{2}}{2\cos \phi}),\ d_{3}\in (0,\frac{(\xi_{1}-\xi_{2})\tan \phi}{2}), \notag\\
&d_{4}\in (0,\frac{\xi_{2}\tan \phi}{2}),\ \widetilde{d_{1}}\in (0,\frac{\vert\xi_{3}\vert}{2\cos \phi}),\  \widetilde{d_{4}}\in (0,\frac{\vert\xi_{3}\vert\tan \phi}{2}),\nonumber
\end{align}
where  $\phi=\phi(\xi)$ is a sufficiently small angle such that $0<\phi<\frac{\pi}{4}$ and the above rays all fall into
their decaying regions. Then, the complex plane $\mathbb{C}$ is separated into eighteen open sectors $\Omega_{ij}(i=0,1,2,j=1\cdots,4),$ and $\Omega_{i}, \overline{\Omega_{i}}(i=3,4,5)$ see
Fig. 6.

Next, we open the jump contour  off the real axis via continuous extensions of the jump matrix $J^{(3)}(z)$.
\begin{prop}\label{prop10} Let $q_{0}\in \tanh x+ H^{4,4}(\mathbb{R})$,
there exist functions $R_{kj}\rightarrow \mathbb{C},k=0,1,2, j=1,\cdots, 4,$ and $R_{k}, \overline{R_{k}}, k=3,4$ such that
\begin{align}
&R_{01}=\begin{cases}
\frac{\overline{r(z)}}{1-\mid r(z)\mid^{2}}\delta_{+}^{-2}(z) \qquad z\in I_{01},\\
0 \qquad z\in \gamma_{01},\\
\end{cases}\
R_{02}=\begin{cases}
\frac{r(z)}{1-\mid r(z)\mid^{2}}\delta_{-}^{2}(z) \qquad z\in I_{02},\\
0 \qquad z\in \gamma_{02},\\
\end{cases}\notag\\
&
R_{03}=\begin{cases}
\frac{r(z)}{1-\mid r(z)\mid^{2}}\delta_{-}^{2}(z) \qquad z\in I_{03},\\
0 \qquad z\in \gamma_{03},\\
\end{cases}\ R_{04}=\begin{cases}
\frac{\overline{r(z)}}{1-\mid r(z)\mid^{2}}\delta_{+}^{-2}(z) \qquad z\in I_{04},\\
0 \qquad z\in \gamma_{04},\\
\end{cases}\notag\\
&k=1,2:\notag\\ &R_{k1}=\begin{cases}
\frac{\overline{r(z)}}{1-\mid r(z)\mid^{2}}\delta_{+}^{-2}(z) \ z\in I_{k1},\\
\frac{\overline{r(\xi_{k})}}{1-\mid r(\xi_{k})\mid^{2}}\delta_{0}^{-2}(\xi_{k})(z-\xi_{k})^{-2i\nu(\xi_{k})\epsilon_{k}} \ z\in \gamma_{k1},\\
\end{cases}\
R_{k2}=\begin{cases}
\frac{r(z)}{1-\mid r(z)\mid^{2}}\delta_{-}^{2}(z) \ z\in I_{k2},\\
\frac{r(\xi_{k})}{1-\mid r(\xi_{k})\mid^{2}}\delta_{0}^{2}(\xi_{k})(z-\xi_{k})^{2i\nu(\xi_{k})\epsilon_{k}} \ z\in \gamma_{k2},\\
\end{cases}\notag\\
&R_{k3}=\begin{cases}
\overline{r(z)}\delta^{-2}(z) \ z\in I_{k3},\\
\overline{r(\xi_{k})}\delta_{0}^{-2}(\xi_{k})(z-\xi_{k})^{-2i\nu(\xi_{k})\epsilon_{k}} \ z\in \gamma_{k3},\\
\end{cases}\
R_{k4}=\begin{cases}
r(z)\delta^{2}(z) \ z\in I_{k4},\\
r(\xi_{k})\delta_{0}^{2}(\xi_{k})(z-\xi_{k})^{2i\nu(\xi_{k})\epsilon_{k}} \ z\in \gamma_{k4},\\
\end{cases}\notag\\
&k=3,4:\notag\\ &R_{k}=\begin{cases}
\frac{\overline{r(z)}}{1-\mid r(z)\mid^{2}}\delta_{+}^{-2}(z) \qquad z\in I_{k1},\\
\frac{\overline{r(\xi_{k})}}{1-\mid r(\xi_{k})\mid^{2}}\delta_{+}^{-2}(\xi_{k})  \qquad z\in \gamma_{k},\\
\end{cases}\
\overline{R_{k}}=\begin{cases}
\frac{r(z)}{1-\mid r(z)\mid^{2}}\delta_{-}^{2}(z) \qquad z\in I_{k2},\\
\frac{r(\xi_{k})}{1-\mid r(\xi_{k})\mid^{2}}\delta_{-}^{2}(\xi_{k}) \qquad z\in \overline{\gamma_{k}},\\
\end{cases}\nonumber
\end{align}
and we have the following estimate
\begin{align}
&k=0,1,2:\notag\\
&\left\vert\bar{\partial} R_{kj}\right\vert\lesssim \left\vert z-\xi_{k}\right\vert^{-1/2}+\left\vert p_{kj}'(Re z)\right\vert,\notag\\
&k=3,4:\notag\\
&\left\vert\bar{\partial} R_{k}\right\vert\lesssim \left\vert z-\xi_{k}\right\vert^{-1/2}+\left\vert p_{k}'(Re z)\right\vert,\nonumber \end{align}
where $<Re z>=\sqrt{1+(Re z)^{2}}, p_{k}=p_{k1}= \frac{\overline{r(z)}}{1-\mid r(z)\mid^{2}}, p_{k2}=\frac{r(z)}{1-\mid r(z)\mid^{2}}, p_{k3}=\overline{r(z)}, p_{k4}=r(z)$.
\end{prop}
\centerline{\begin{tikzpicture}[scale=0.8]
\draw[-][thick](5,-2)--(5,-1);
\draw[->][thick](5,0)--(5,-1);
\draw[->][thick](5,2)--(5,1);
\draw[-][thick](5,0)--(5,1);
\draw[-][thick](2,-1)--(2,-0.5);
\draw[->][thick](2,0)--(2,-0.5);
\draw[->][thick](2,1)--(2,0.5);
\draw[-][thick](2,0)--(2,0.5);
\draw[-][thick](0,-1)--(0,-0.5);
\draw[->][thick](0,0)--(0,-0.5);
\draw[->][thick](0,1)--(0,0.5);
\draw[-][thick](0,0)--(0,0.5);
\draw[-][thick](-5,0)--(-6,-1);
\draw[<-][thick](-6,-1)--(-7,-2);
\draw[-][thick](-5,0)--(-6,1);
\draw[<-][thick](-6,1)--(-7,2);
\draw[->][thick](-3,0)--(-2,0);
\draw[-][dashed,thick](-7,0)--(-6,0);
\draw[-][dashed,thick](-6,0)--(-5,0);
\draw[-][thick](-5,0)--(-3,0);
\draw[-][thick](-2.0,0)--(-1.0,0)node[below]{$\xi_{3}$};
\draw[-][dashed,thick](-1,0)--(0,0);
\draw[->][thick](0,1)--(0.5,0.5);
\draw[-][thick](0.5,0.5)--(1,0);
\draw[->][thick](0,-1)--(0.5,-0.5);
\draw[-][thick](0.5,-0.5)--(1,0);
\draw[-][thick](2,-1)--(1.5,-0.5);
\draw[<-][thick](1.5,-0.5)--(1,0);
\draw[-][thick](2,1)--(1.5,0.5);
\draw[<-][thick](1.5,0.5)--(1,0);
\draw[-][thick](0,1)--(-0.5,0.5);
\draw[->][thick](-1,0)--(-0.5,0.5);
\draw[-][thick](0,-1)--(-0.5,-0.5);
\draw[->][thick](-1,0)--(-0.5,-0.5);
\draw[-][dashed,thick](0,0)--(1,0)node[below]{$0$};
\draw[-][dashed,thick](1,0)--(2,0);
\draw[fill] (1,0) circle [radius=0.035];
\draw[fill] (3,0) circle [radius=0.035]node[below]{$\xi_{2}$};
\draw[fill] (7,0) circle [radius=0.035]node[below]{$\xi_{1}$};
\draw[fill] (-5,0) circle [radius=0.035]node[below]{$\xi_{4}$};
\draw[fill] (-3,0) circle [radius=0.035]node[below]{$-1$};
\draw[-][dashed,thick](3,0)--(4,0);
\draw[-][dashed,thick](4,0)--(5,0);
\draw[-][dashed,thick](5,0)--(7,0);
\draw[-][dashed,thick](7,0)--(8,0);
\draw[-][dashed,thick](8,0)--(9,0);
\draw[fill] (-1,0) circle [radius=0.035];
\draw[-][dashed,thick](2.0,0)--(3.0,0);
\draw[-][thick](4,1)--(5,2);
\draw[->][thick](5,2)--(6,1);
\draw[-][thick](6,1)--(7,0);
\draw[->][thick](7,0)--(8,1);
\draw[-][thick](8,1)--(9,2);
\draw[-][thick](4,-1)--(5,-2);
\draw[->][thick](5,-2)--(6,-1);
\draw[-][thick](6,-1)--(7,0);
\draw[->][thick](7,0)--(8,-1);
\draw[-][thick](8,-1)--(9,-2);
\draw[fill] [red](-3,1.0) node{$\Omega_{5}$};
\draw[fill] [red](-3,-1.0) node{$\overline{\Omega_{5}}$};
\draw[fill] [red](-0.3,0.25) node{$\Omega_{3}$};
\draw[fill] [red](-0.3,-0.25) node{$\overline{\Omega_{3}}$};
\draw[fill] [red](0.4,0.25) node{$\Omega_{01}$};
\draw[fill] [red](0.4,-0.25) node{$\Omega_{02}$};
\draw[fill] [red](1.65,0.25) node{$\Omega_{04}$};
\draw[fill] [red](1.65,-0.25) node{$\Omega_{03}$};
\draw[fill] [red](2.35,-0.25) node{$\Omega_{22}$};
\draw[fill] [red](2.35,0.25) node{$\Omega_{21}$};
\draw[fill] [red](4,0.5) node{$\Omega_{24}$};
\draw[fill] [red](4,-0.5) node{$\Omega_{23}$};
\draw[fill] [red](-6,0.5) node{$\Omega_{4}$};
\draw[fill] [red](-6,-0.5) node{$\overline{\Omega_{4}}$};
\draw[fill] [red](6,-0.5) node{$\Omega_{13}$};
\draw[fill] [red](6,0.5) node{$\Omega_{14}$};
\draw[fill] [red](8,0.5) node{$\Omega_{11}$};
\draw[fill] [red](8,-0.5) node{$\Omega_{12}$};
\draw[->][thick](2,1)--(2.5,0.5);
\draw[-][thick](2.5,0.5)--(3,0);
\draw[->][thick](2,-1)--(2.5,-0.5);
\draw[-][thick](2.5,-0.5)--(3,0);
\draw[-][thick](3,0)--(3.5,0.5);
\draw[->][thick](3.5,0.5)--(4,1);
\draw[-][thick](3,0)--(3.5,-0.5);
\draw[->][thick](3.5,-0.5)--(4,-1);
\draw[fill][cyan] (-0.7,0.7) node{$\gamma_{3}$};
\draw[fill] [cyan](-0.7,-0.7) node{$\overline{\gamma_{3}}$};
\draw[fill] [cyan](0.7,0.7) node{$\gamma_{01}$};
\draw[fill] [cyan](0.7,-0.7) node{$\gamma_{02}$};
\draw[fill] [cyan](1.3,0.7) node{$\gamma_{04}$};
\draw[fill][cyan] (1.3,-0.7) node{$\gamma_{03}$};
\draw[fill] [cyan](2.7,0.7) node{$\gamma_{21}$};
\draw[fill] [cyan](2.7,-0.7) node{$\gamma_{22}$};
\draw[fill] [cyan](4.0,1.3) node{$\gamma_{24}$};
\draw[fill][cyan] (4.0,-1.3) node{$\gamma_{23}$};
\draw[fill] [cyan](-6,1.4) node{$\gamma_{4}$};
\draw[fill] [cyan](-6,-1.4) node{$\overline{\gamma_{4}}$};
\draw[fill] [cyan](6,1.4) node{$\gamma_{14}$};
\draw[fill] [cyan](6,-1.4) node{$\gamma_{13}$};
\draw[fill] [cyan](8,1.4) node{$\gamma_{11}$};
\draw[fill] [cyan](8,-1.4) node{$\gamma_{12}$};
\end{tikzpicture}}
\centerline{\noindent {\small \textbf{Figure 6.} (Color online) The jump contour $\Sigma^{(4)}$.}}

Then, we introduce a transformation
\begin{align}\label{107}
M^{(4)}=M^{(3)}R^{(3)},
\end{align}
where
\begin{align}
R^{(3)}=\begin{cases}
\left(\begin{array}{cc}
    1  &  R_{01}e^{2itf}\\
    0  &  1\\
\end{array}\right),\quad z\in\Omega_{01}\cup\Omega_{04},\
\left(\begin{array}{cc}
    1  &  0\\
    R_{02}e^{-2itf}  &  1\\
\end{array}\right),\quad z\in\Omega_{02}\cup\Omega_{03},\\
k=1,2:\\
\left(\begin{array}{cc}
    1  &  R_{k1}e^{2itf}\\
    0  &  1\\
\end{array}\right),\qquad z\in\Omega_{k1},\qquad
\left(\begin{array}{cc}
    1  &  0\\
    R_{k2}e^{-2itf}  &  1\\
\end{array}\right),\qquad z\in\Omega_{k2},\\
\left(\begin{array}{cc}
    1  &  -R_{k3}e^{2itf}\\
    0  &  1\\
\end{array}\right),\qquad z\in\Omega_{k3},\qquad
\left(\begin{array}{cc}
    1  &  0\\
    -R_{k4}e^{-2itf}  &  1\\
\end{array}\right),\qquad z\in\Omega_{k4},\\
k=3,4:\\
\left(\begin{array}{cc}
    1  &  R_{k}e^{2itf}\\
    0  &  1\\
\end{array}\right),\qquad z\in\Omega_{k},\qquad
\left(\begin{array}{cc}
    1  &  0\\
    \overline{R_{k}}e^{-2itf}  &  1\\
\end{array}\right),\qquad z\in\overline{\Omega_{k}},\\
\left(\begin{array}{cc}
    1  &  0\\
    0  &  1\\
\end{array}\right),\qquad z\in\Omega_{5}\cup\overline{\Omega_{5}}.
\end{cases}\nonumber
\end{align}

Then, combining the RH problem \ref{rhp10} and Proposition \ref{prop10}, we can immediately
obtain the following $\bar{\partial}$-RH problem $M^{(4)}$.

\begin{rhp}\label{rhp11} Find an analysis function $M^{(4)}(x, t, z)$ with the
following properties:\\

$\bullet$ $M^{(4)}(x, t, z)$ is analytical in $\mathbb{C}\setminus \Sigma^{(4)}$;\\

$\bullet$ $\bar{\partial}M^{(4)}=M^{(4)}\bar{\partial}R^{(3)}(z)$, as $\lambda\in\mathbb{C}\setminus \Sigma^{(4)}$, where
\begin{align}
\bar{\partial}R^{(3)}=\begin{cases}
\left(\begin{array}{cc}
    1  &  \bar{\partial}R_{01}e^{2itf}\\
    0  &  1\\
\end{array}\right),\quad z\in\Omega_{01}\cup\Omega_{04},\
\left(\begin{array}{cc}
    1  &  0\\
    \bar{\partial}R_{02}e^{-2itf}  &  1\\
\end{array}\right),\quad z\in\Omega_{02}\cup\Omega_{03},\\
k=1,2:\\
\left(\begin{array}{cc}
    1  &  \bar{\partial}R_{k1}e^{2itf}\\
    0  &  1\\
\end{array}\right),\qquad z\in\Omega_{k1},\qquad
\left(\begin{array}{cc}
    1  &  0\\
    \bar{\partial}R_{k2}e^{-2itf}  &  1\\
\end{array}\right),\qquad z\in\Omega_{k2},\\
\left(\begin{array}{cc}
    1  &  -\bar{\partial}R_{k3}e^{2itf}\\
    0  &  1\\
\end{array}\right),\qquad z\in\Omega_{k3},\qquad
\left(\begin{array}{cc}
    1  &  0\\
    -\bar{\partial}R_{k4}e^{-2itf}  &  1\\
\end{array}\right),\qquad z\in\Omega_{k4},\\
k=3,4:\\
\left(\begin{array}{cc}
    1  &  \bar{\partial}R_{k}e^{2itf}\\
    0  &  1\\
\end{array}\right),\qquad z\in\Omega_{k},\qquad
\left(\begin{array}{cc}
    1  &  0\\
    \bar{\partial}\overline{R_{k}}e^{-2itf}  &  1\\
\end{array}\right),\qquad z\in\overline{\Omega_{k}};
\end{cases}\nonumber
\end{align}

$\bullet$ $M_{+}^{(4)}(x, t, z)=M^{(4)}_{-}(x, t, z)J^{(4)}(x, t, z), z\in \Sigma^{(4)}$, where
\begin{align}\label{109}
J^{(4)}(x, t, z)=\begin{cases}
\left(\begin{array}{cc}
    1  &  0\\
    \frac{r(z)e^{-2itf}\delta_{-}^{2}}{1-\left\vert r(z) \right\vert^{2}}  &  1\\
\end{array}\right)\left(\begin{array}{cc}
    1  &  -\frac{\overline{r(z)}e^{2itf}\delta_{+}^{-2}}{1-\mid r(z) \mid^{2}}\\
    0  &  1\\
\end{array}\right),\qquad z\in (\xi_{4},\xi_{3}),\\
\left(\begin{array}{cc}
    1  &  -R_{3}e^{2ift}\\
    0  &  1\\
\end{array}\right),\ z\in \gamma_{3+},\qquad
\left(\begin{array}{cc}
    1  &  0\\
    -\overline{R_{3}}e^{-2ift}  &  1\\
\end{array}\right),\ z\in \gamma_{3-},\\
\left(\begin{array}{cc}
    1  &  R_{21}e^{2ift}\\
    0  &  1\\
\end{array}\right),\ z\in \gamma_{2+},\qquad
\left(\begin{array}{cc}
    1  &  0\\
    R_{22}e^{-2ift}  &  1\\
\end{array}\right),\qquad z\in \gamma_{2-},\\
\left(\begin{array}{cc}
    1  &  0\\
    (R_{24}-R_{14})e^{-2ift}  &  1\\
\end{array}\right),\ z\in \gamma_{1+},\
\left(\begin{array}{cc}
    1  &  (R_{23}-R_{13})e^{2ift}\\
    0  &  1\\
\end{array}\right),\ z\in \gamma_{1-},\\
k=1,2:\\
\left(\begin{array}{cc}
    1  &  -R_{k1}e^{2itf}\\
    0  &  1\\
\end{array}\right),\ z\in\gamma_{k1},\qquad
\left(\begin{array}{cc}
    1  &  0\\
    R_{k2}e^{-2itf}  &  1\\
\end{array}\right),\ z\in\gamma_{k2},\\
\left(\begin{array}{cc}
    1  &  -R_{k3}e^{2itf}\\
    0  &  1\\
\end{array}\right),\ z\in\gamma_{k3},\qquad
\left(\begin{array}{cc}
    1  &  0\\
    R_{k4}e^{-2itf}  &  1\\
\end{array}\right),\ z\in\gamma_{k4},\\
k=3,4:\\
\left(\begin{array}{cc}
    1  &  -R_{k}e^{2itf}\\
    0  &  1\\
\end{array}\right),\ z\in \gamma_{k},\qquad
\left(\begin{array}{cc}
    1  &  0\\
    \overline{R_{k}}e^{-2itf}  &  1\\
\end{array}\right),\ z\in \overline{\gamma_{k}};
\end{cases}
\end{align}

$\bullet$ $M^{(4)}(x, t, z)=\mathbb{I}+O(\frac{1}{z})$, as $z\rightarrow\infty$.
\end{rhp}

\subsection{The decomposition of the mixed $\bar{\partial}$-RH problem }
We decompose $M^{(4)}(z)$ into a pure RH problem $M^{RHP}(z)$ with $\bar{\partial}R^{(3)}(z)=0$ and a pure $\bar{\partial}-$problem $M^{(5)}(z)$ with $\bar{\partial}R^{(3)}(z)\neq0$ in the form
\begin{align}\label{111}
M^{(4)}(z)=M^{(5)}(z)M^{RHP}(z).
\end{align}

\subsubsection{Pure RH problem}
In the case of $\bar{\partial}R^{(2)}=0$, a pure  RH problem is constructed as follows:

\begin{rhp}\label{rhp12}
 Find an analysis function $M^{RHP}(x, t, z)$ with the
following properties:\\

$\bullet$ $M^{RHP}(x, t, z)$ is meromorphic in $\mathbb{C}\setminus \Sigma^{(4)}$;\\

$\bullet$ $M_{+}^{RHP}(x, t, z)=M^{RHP}_{-}(x, t, z)J^{(4)}(x, t, z), z\in \Sigma^{(4)}$, where $J^{(4)}(x, t, z)$ is given in \eqref{109};

$\bullet$ $M^{RHP}(x, t, z)=\mathbb{I}+O(\frac{1}{z})$, as $z\rightarrow\infty$.
\end{rhp}

Define three open disks
\begin{align}
\mathcal{U}_{-1}=\left\{z:\left\vert z+1\right\vert\leq (\frac{15}{4}t)^{-1/3}\varepsilon\right\} ,\qquad \mathcal{U}_{\xi_{k}}=\{z:\left\vert z-\xi_{k}\right\vert\leq \rho\},\ k=1,2,\nonumber
\end{align}
where $\rho$ has been defined in \eqref{38}, $\varepsilon$ is a constant admitting $\varepsilon\leq\sqrt{2C}$  and $\sqrt{2C}(\frac{15}{4}t)^{-1/3}<\rho$ for $t$ being large enough.
Introducing a matrix RH problem $M^{RHP}(x, t, z)$ as follows
\begin{align}\label{113}
M^{RHP}(z)=\begin{cases}
E(z),\qquad\qquad  z\in \mathbb{C}\backslash \mathcal{U}_{-1}\cup\mathcal{U}_{\xi_{1}}\cup\mathcal{U}_{\xi_{2}},\\
E(z)M^{(-1)},\qquad z\in \mathcal{U}_{-1},\\
E(z)M^{(\xi_{1})},\qquad z\in \mathcal{U}_{\xi_{1}},\\
E(z)M^{(\xi_{2})},\qquad z\in \mathcal{U}_{\xi_{2}},\\
\end{cases}
 \end{align}
where $M^{(-1)}$ is a localized model near $z=-1$,
which can be solved by the Painlev\'{e} II equation. $M^{(\xi_{k})}$ are the known parabolic cylinder model, which have been solved in Appendix \ref{A},
and $E(z)$, an error function, is the solution of a small-norm RH problem.

Next, we will consider the local paramatrix near the  phase points $-1$, $\xi_{2}$ and $\xi_{1}$.
For $z$ near $-1$, the phase faction $tf(z)$ can be approximated with the following scaled spectral variables:
\begin{align}
tf(z)=\frac{4}{3}\hat{k}^{3}+s\hat{k}+O(\hat{k}^{4}t^{-1/3}),\nonumber
\end{align}
where the scaled parameters
\begin{align}\label{115}
\hat{k}=(\frac{15}{4}t)^{1/3}(z+1),\qquad s=\frac{4}{15}(8+\xi)(\frac{15}{4}t)^{2/3}.
\end{align}
For $z$ near $\xi_{1}$ and $\xi_{2}$, the scaled spectral variables are the same as \eqref{60}.

\begin{prop}\label{prop11}
In the transition region, the two scaled phase points $\hat{k}_{j}$ are always within a fixed interval, given by
\begin{align}
\hat{k}_{j}\in \left(-\sqrt{2C}(\frac{15}{4})^{1/3}, \sqrt{2C}(\frac{15}{4})^{1/3}\right),\quad j=3,4.\nonumber
 \end{align}
 where $\hat{k}_{j}=(\frac{15}{4}t)^{1/3}(\xi_{j}+1), j=3, 4,$
\end{prop}

The transformation defined by \eqref{115}
 maps $\mathcal{U}_{-1}$ onto the disk $\mathcal{U}_{0}$ where $\mathcal{U}_{0}=\left\{\hat{k}:\left\vert \hat{k}\right\vert\leq \varepsilon\right\}$ (see Fig. 7). Proposition \ref{prop11} indicates that $\hat{k}_{3},\hat{k}_{4} \in\mathcal{U}_{0}$ for large $t$.\\
\centerline{\begin{tikzpicture}[scale=0.8]
\draw[-][dashed,thick](-3,0)--(-1,0);
\draw[-][dashed,thick](1,0)--(3,0);
\draw[-][thick](-1,0)--(0,0) circle [radius=0.035]node[below]{$-1$};
\draw[->][thick](0,0)--(0.5,0);
\draw[-][thick](0.5,0)--(1,0);
\draw[-][thick](-2,1)--(-1,0);
\draw[->][thick](-3,2)--(-2,1);
\draw[-][thick](-2,-1)--(-1,0);
\draw[->][thick](-3,-2)--(-2,-1);
\draw[<-][thick](2,1)--(1,0);
\draw[-][thick](3,2)--(2,1);
\draw[<-][thick](2,-1)--(1,0);
\draw[-][thick](3,-2)--(2,-1);
\draw[fill] (-1,0) node[below]{$\xi_{4}$};
\draw[fill] (1,0) node[below]{$\xi_{3}$};
\draw[fill] [cyan](-2,1.3) node{$\gamma_{4}$};
\draw[fill] [cyan](-2,-1.4) node{$\overline{\gamma_{4}}$};
\draw[fill] [cyan](2,1.3) node{$\gamma_{3}$};
\draw[fill] [cyan](2,-1.4) node{$\overline{\gamma_{3}}$};
\draw[fill] [red](-2.5,0.5) node{$\Omega_{4}$};
\draw[fill] [red](-2.5,-0.5) node{$\overline{\Omega_{4}}$};
\draw[fill] [red](2.5,0.5) node{$\Omega_{3}$};
\draw[fill] [red](2.5,-0.5) node{$\overline{\Omega_{3}}$};
\draw[fill] [cyan](0,-2) node{$\partial\mathcal{U}_{-1}$};
\draw[-][dashed,thick](1.5,0) arc(0:360:1.5);
\draw[-][dashed,thick](1.5,0) arc(0:30:1.5);
\draw[->][dashed,thick](1.5,0) arc(0:150:1.5);
\draw[-][dashed,thick](1.5,0) arc(0:210:1.5);
\draw[-][dashed,thick](1.5,0) arc(0:330:1.5);
\draw[-][dashed,thick](4,0)--(6,0);
\draw[-][dashed,thick](8,0)--(10,0);
\draw[-][thick](6,0)--(7,0) circle [radius=0.035]node[below]{$0$};
\draw[->][thick](7,0)--(7.5,0);
\draw[-][thick](7.5,0)--(8,0);
\draw[-][thick](5,1)--(6,0);
\draw[->][thick](4,2)--(5,1);
\draw[-][thick](5,-1)--(6,0);
\draw[->][thick](4,-2)--(5,-1);
\draw[<-][thick](9,1)--(8,0);
\draw[-][thick](10,2)--(9,1);
\draw[<-][thick](9,-1)--(8,0);
\draw[-][thick](10,-2)--(9,-1);
\draw[fill] (6,0) node[below]{$\hat{k}_{4}$};
\draw[fill] (8,0) node[below]{$\hat{k}_{3}$};
\draw[fill] [cyan](5,1.3) node{$\hat{\gamma_{4}}$};
\draw[fill] [cyan](5,-1.4) node{$\hat{\overline{\gamma_{4}}}$};
\draw[fill] [cyan](9,1.3) node{$\hat{\gamma_{3}}$};
\draw[fill] [cyan](9,-1.4) node{$\hat{\overline{\gamma_{3}}}$};
\draw[fill] [red](4.5,0.5) node{$\hat{\Omega_{4}}$};
\draw[fill] [red](4.5,-0.5) node{$\hat{\overline{\Omega_{4}}}$};
\draw[fill] [red](9.5,0.5) node{$\hat{\Omega_{3}}$};
\draw[fill] [red](9.5,-0.5) node{$\hat{\overline{\Omega_{3}}}$};
\draw[fill] [cyan](7,-2) node{$\partial\mathcal{U}_{0}$};
\draw[-][dashed,thick](8.5,0) arc(0:360:1.5);
\draw[-][dashed,thick](8.5,0) arc(0:30:1.5);
\draw[->][dashed,thick](8.5,0) arc(0:150:1.5);
\draw[-][dashed,thick](8.5,0) arc(0:210:1.5);
\draw[-][dashed,thick](8.5,0) arc(0:330:1.5);
\end{tikzpicture}}
\centerline{\noindent {\small \textbf{Figure 7.} (Color online) The map relation between two disks $\mathcal{U}_{-1}$ and $\mathcal{U}_{0}$.}}

\begin{rhp}\label{rhp13}
The analysis function $M^{(-1)}(z)$ has the
following properties:\\

 $\bullet$ $M^{(-1)}(z)$ is meromorphic in $\mathcal{U}_{-1}\setminus \Sigma_{-1}$, where $\Sigma_{-1}=\Sigma^{(4)}\cap\mathcal{U}_{-1}$;\\

 $\bullet$ $M_{+}^{(-1)}(z)=M^{(-1)}_{-}(z)J^{(-1)}(z), z\in \Sigma_{-1}$, where
\begin{align}\label{118}
J^{(-1)}(z)=
\begin{cases}
\left(\begin{array}{cc}
    1  &  0\\
    \frac{r(z)e^{-2itf}\delta_{-}^{2}}{1-\left\vert r(z) \right\vert^{2}}  &  1\\
\end{array}\right)\left(\begin{array}{cc}
    1  &  -\frac{\overline{r(z)}e^{2itf}\delta_{+}^{-2}}{1-\mid r(z) \mid^{2}}\\
    0  &  1\\
\end{array}\right),\qquad z\in (\xi_{4},\xi_{3}),\\
\left(\begin{array}{cc}
    1  &  -R_{k}e^{2itf}\\
    0  &  1\\
\end{array}\right),\qquad z\in \gamma_{k},\quad k=3,4,\\
\left(\begin{array}{cc}
    1  &  0\\
    \overline{R_{k}}e^{-2itf}  &  1\\
\end{array}\right),\qquad z\in \overline{\gamma_{k}}, \quad k=3,4;
\end{cases}
\end{align}

$\bullet$ $M^{(-1)}(z)(\mathbb{I}+O(t^{-1/3}))^{-1}\rightarrow\mathbb{I}$ as $t\rightarrow\infty$, uniformly for $z\in\partial\mathcal{U}_{-1}$.
\end{rhp}

Denote
\begin{align}\label{119}
R(z)=\frac{\overline{r(z)}}{1-\mid r(z)\mid^{2}}\delta_{+}^{-2}(z),
\end{align}
the following Proposition is given out.
\begin{prop}\label{prop12}
Let $q_{0}\in \tanh(x)+ L^{1,2}(\mathbb{R}), q'_{0}\in W^{1,\infty}(\mathbb{R}),$ and $-C<(\xi+8)t^{2/3}<0$, then we have
\begin{align}
&\vert R(z)e^{2ift(z)}-R(-1)e^{\frac{8}{3}i\hat{k}^{3}+2is\hat{k}}\vert\lesssim t^{-1/6}, \qquad \hat{k}\in (\hat{k}_{4},\hat{k}_{3}),\notag\\
&\vert R(\xi_{j})e^{2ift(z)}-R(-1)e^{\frac{8}{3}i\hat{k}^{3}+2is\hat{k}}\vert\lesssim t^{-1/6}, \qquad \hat{k}\in \hat{\gamma_{j}}\cup\hat{\overline{\gamma_{j}}},\quad j=3,4.\nonumber
\end{align}
\end{prop}
In terms of the above proposition, we obtain the following RH problem $M^{(0)}(\hat{k})$:
\begin{rhp}\label{rhp14}
The analysis function $M^{(0)}(\hat{k})$ has the
following properties:\\

 $\bullet$ $M^{(0)}(\hat{k})$ is meromorphic in $\mathcal{U}_{0}\setminus \Sigma_{0}$, where $\Sigma_{0}=\Sigma^{(4)}\cap\mathcal{U}_{0}$;\\

 $\bullet$ $M_{+}^{(0)}(\hat{k})=M^{(0)}_{-}(\hat{k})J^{(0)}(\hat{k}), \hat{k}\in \Sigma_{0}$, where
\begin{align}\label{121}
J^{(0)}(\hat{k})=
\begin{cases}
b_{-}b_{+},\qquad \hat{k}\in (\hat{k}_{4},\hat{k}_{3}),\\
\left(\begin{array}{cc}
    1  &  -R(-1)e^{\frac{8}{3}i\hat{k}^{3}+2is\hat{k}}\\
    0  &  1\\
\end{array}\right)=b_{+},\qquad \hat{k}\in \hat{\gamma_{k}},\quad k=3,4,\\
\left(\begin{array}{cc}
    1  &  0\\
    \overline{R(-1)}e^{-(\frac{8}{3}i\hat{k}^{3}+2is\hat{k})}  &  1\\
\end{array}\right)=b_{-},\qquad \hat{k}\in \hat{\overline{\gamma_{k}}}, \quad k=3,4;
\end{cases}
\end{align}

$\bullet$ $M^{(0)}(\hat{k})\rightarrow\mathbb{I}$ as $\hat{k}\rightarrow\infty$.
\end{rhp}

Furthermore, the following result can be arrived.
\begin{prop}\label{prop13}
Let  $-C<(\xi+8)t^{2/3}<0$, then for large $t$, we have
\begin{align}\label{122}
&J^{(-1)}(z)=J^{(0)}(\hat{k})+O(t^{-1/6}),\qquad \hat{k}\in\mathcal{U}_{0},\notag\\
&M^{(-1)}(z)=M^{(0)}(\hat{k})+O(t^{-1/6}),\qquad \hat{k}\in\mathcal{U}_{0}.
\end{align}
\end{prop}
Therefore the solution of $M^{(0)}$ is crucial to our analysis, and it can be converted to the standard Painlev\'{e} II equation by proper deformation  as shown in ``Appendix \ref{B}".\\
\centerline{\begin{tikzpicture}[scale=0.9]
\draw[<-][thick](-0.5,0) arc(0:360:0.5);
\draw[-][thick](-0.5,0) arc(0:30:0.5);
\draw[-][thick](-0.5,0) arc(0:150:0.5);
\draw[-][thick](-0.5,0) arc(0:210:0.5);
\draw[-][thick](-0.5,0) arc(0:330:0.5);
\draw[<-][thick](3.5,0) arc(0:360:0.5);
\draw[-][thick](3.5,0) arc(0:30:0.5);
\draw[-][thick](3.5,0) arc(0:150:0.5);
\draw[-][thick](3.5,0) arc(0:210:0.5);
\draw[-][thick](3.5,0) arc(0:330:0.5);
\draw[<-][thick](8,0) arc(0:360:1);
\draw[-][thick](8,0) arc(0:30:1);
\draw[-][thick](8,0) arc(0:150:1);
\draw[-][thick](8,0) arc(0:210:1);
\draw[-][thick](8,0) arc(0:330:1);
\draw[-][thick](5,-2)--(5,-1);
\draw[->][thick](5,0)--(5,-1);
\draw[->][thick](5,2)--(5,1);
\draw[-][thick](5,0)--(5,1);
\draw[-][thick](2,-1)--(2,-0.5);
\draw[->][thick](2,0)--(2,-0.5);
\draw[->][thick](2,1)--(2,0.5);
\draw[-][thick](2,0)--(2,0.5);
\draw[-][thick](0,-1)--(0,-0.5);
\draw[->][thick](0,0)--(0,-0.5);
\draw[->][thick](0,1)--(0,0.5);
\draw[-][thick](0,0)--(0,0.5);
\draw[<-][thick](-2,-1)--(-3,-2);
\draw[<-][thick](-2,1)--(-3,2);
\draw[->][thick](0,1)--(0.5,0.5);
\draw[-][thick](0.5,0.5)--(1,0);
\draw[->][thick](0,-1)--(0.5,-0.5);
\draw[-][thick](0.5,-0.5)--(1,0);
\draw[-][thick](2,-1)--(1.5,-0.5);
\draw[<-][thick](1.5,-0.5)--(1,0);
\draw[-][thick](2,1)--(1.5,0.5);
\draw[<-][thick](1.5,0.5)--(1,0);
\draw[-][thick](0,1)--(-0.5,0.5);
\draw[->][thick](-0.65,0.35)--(-0.5,0.5);
\draw[-][thick](0,-1)--(-0.5,-0.5);
\draw[->][thick](-0.65,-0.35)--(-0.5,-0.5);
\draw[-][thick](-1.35,0.35)--(-1.5,0.5);
\draw[-][thick](-1.5,0.5)--(-2,1);
\draw[-][thick](-1.35,-0.35)--(-1.5,-0.5);
\draw[-][thick](-1.5,-0.5)--(-2,-1);
\draw[fill] (1,0) circle [radius=0.035] node[below]{$0$};
\draw[-][thick](4,1)--(5,2);
\draw[->][thick](5,2)--(6,1);
\draw[-][thick](6,1)--(6.3,0.7);
\draw[->][thick](7.7,0.7)--(8,1);
\draw[-][thick](8,1)--(9,2);
\draw[-][thick](4,-1)--(5,-2);
\draw[->][thick](5,-2)--(6,-1);
\draw[-][thick](6,-1)--(6.3,-0.7);
\draw[->][thick](7.7,-0.7)--(8,-1);
\draw[-][thick](8,-1)--(9,-2);
\draw[->][thick](2,1)--(2.5,0.5);
\draw[-][thick](2.5,0.5)--(2.65,0.35);
\draw[->][thick](2,-1)--(2.5,-0.5);
\draw[-][thick](2.5,-0.5)--(2.65,-0.35);
\draw[-][thick](3.35,0.35)--(3.5,0.5);
\draw[->][thick](3.5,0.5)--(4,1);
\draw[-][thick](3.35,-0.35)--(3.5,-0.5);
\draw[->][thick](3.5,-0.5)--(4,-1);
\draw[fill] [cyan] (-1,-0.9) node{$\partial\mathcal{U}_{-1}$};
\draw[fill] [cyan] (3,-0.9) node{$\partial\mathcal{U}_{\xi_{2}}$};
\draw[fill] [cyan] (7,-1.4) node{$\partial\mathcal{U}_{\xi_{1}}$};
\end{tikzpicture}}
\centerline{\noindent {\small \textbf{Figure 8.} (Color online) The jump contour $\Sigma^{E}$.}}

Next, we discuss the error function $E(z)$ given by \eqref{113}, which meets the following RH problem.
\begin{rhp}\label{rhp15}
Find a matrix-valued function $E(z)$ has the
following properties:\\

 $\bullet$ $E(z)$ is meromorphic in $\mathbb{C}\setminus \Sigma^{E}$, where $\Sigma^{E}=(\partial\mathcal{U}_{-1}\cup\partial\mathcal{U}_{\xi_{1}}\cup\partial\mathcal{U}_{\xi_{2}})
 \cup(\Sigma^{(4)}\setminus(\mathcal{U}_{-1}\cup\mathcal{U}_{\xi_{1}}\cup\mathcal{U}_{\xi_{2}}))$,  see Fig. 8;\\

 $\bullet$ $E_{+}(z)=E_{-}(z)J^{E}(z), z\in \Sigma^{E}$, where
\begin{align}\label{123}
J^{E}(z)=
\begin{cases}
J^{(4)},\qquad z\in\Sigma^{(3)}\setminus\mathcal{U}_{\xi},\\
M^{(\xi_{1})},\qquad z\in \partial\mathcal{U}_{\xi_{1}},\\
M^{(\xi_{2})},\qquad z\in \partial\mathcal{U}_{\xi_{2}},\\
M^{(-1)},\qquad z\in \partial\mathcal{U}_{-1}.
\end{cases}
\end{align}

$\bullet$ $E(z)=\mathbb{I}+O(\frac{1}{z})$ as $z\rightarrow\infty$.
\end{rhp}

\begin{prop}\label{prop14}
We obtain the following estimates for the jump matrix $J^{E}$  defined in \eqref{123}
\begin{align}
\left\vert J^{E}(z)-\mathbb{I}\right\vert=\begin{cases}
O(e^{-h_{p}t}),\qquad z\in \Sigma^{(4)}\setminus(\mathcal{U}_{\xi_{1}}\cup\mathcal{U}_{\xi_{2}}\cup\mathcal{U}_{-1}),\\
O(t^{-1/3}),\qquad z\in \partial\mathcal{U}_{\xi_{1}}\cup\partial\mathcal{U}_{\xi_{2}}\cup\partial\mathcal{U}_{-1}.
\end{cases}\nonumber
\end{align}
\end{prop}
\begin{proof}
The proof is similar with Proposition \ref{prop3}.
\end{proof}

Define the same Cauchy integral operator as \eqref{72} and \eqref{73}, then we have
\begin{align}
\left\Vert J^{E}-\mathbb{I}\right\Vert_{L^{2}(\Sigma^{E})}=O(t^{-\frac{1}{3}}), \quad
\left\Vert C_{\omega_{E}}\right\Vert_{L^{2}(\Sigma^{E})}\lesssim O(t^{-\frac{1}{3}}),\quad
\left\Vert \mu_{E}-\mathbb{I}\right\Vert_{L^{2}(\Sigma^{E})}\lesssim O(t^{-\frac{1}{3}}).\nonumber
\end{align}
To recover the potential $q(x, t)$, we need the properties of $E(z)$ at $z=0$ and $z=\infty$.
We make the expansion of $E(z)$ at $z=\infty$
\begin{align}\label{126}
E(z)=\mathbb{I}+\frac{E_{1}}{z}+O(z^{-2}), \qquad z\rightarrow \infty,
 \end{align}
where
\begin{align}\label{127}
E_{1}=-\frac{1}{2\pi i}\int_{\Sigma^{E}}\mu_{E}(\zeta)(J^{E}-\mathbb{I})\mathrm{d}\zeta.
 \end{align}
Then, as $t\rightarrow\infty$, the asymptotic behavior of $E_{1}$ can be calculated as
\begin{align}\label{128}
E_{1}&=-\sum_{k=1}^{2}\frac{1}{2\pi i}\oint_{\partial\mathcal{U}_{\xi_{k}}}(J^{E}-\mathbb{I})\mathrm{d}\zeta-\frac{1}{2\pi i}\oint_{\partial\mathcal{U}_{-1}}(J^{E}-\mathbb{I})\mathrm{d}\zeta+O(t^{-2/3})\notag\\
= &-\sum_{k=1}^{2}\frac{1}{2\pi i}\oint_{\partial\mathcal{U}_{\xi_{k}}}\frac{M^{(\xi_{k})}_{1}}{\sqrt{2t\epsilon_{k}f''(\xi_{k})}(\zeta-\xi_{k})}\mathrm{d}\zeta-\frac{1}{2\pi i}\oint_{\partial\mathcal{U}_{-1}}\frac{M^{(-1)}_{1}}{(\frac{15}{4}t)^{1/3}(\zeta+1)}\mathrm{d}\zeta+O(t^{-2/3})\notag\\
=&-\sum_{k=1}^{2}\frac{M^{(\xi_{k})}_{1}}{\sqrt{2t\epsilon_{k}f''(\xi_{k})}}-\frac{M^{(-1)}_{1}}{(\frac{15}{4}t)^{1/3}}+O(t^{-2/3})
=-\frac{M^{(-1)}_{1}}{(\frac{15}{4}t)^{1/3}}+O(t^{-1/2}).
\end{align}
\begin{align}\label{129}
E_{0}&=\mathbb{I}+\sum_{k=1}^{2}\frac{1}{2\pi i}\oint_{\partial\mathcal{U}_{\xi_{k}}}\frac{J^{E}-\mathbb{I}}{\zeta}\mathrm{d}\zeta+\frac{1}{2\pi i}\oint_{\partial\mathcal{U}_{-1}}\frac{J^{E}-\mathbb{I}}{\zeta}\mathrm{d}\zeta+O(t^{-2/3})\notag\\
= &\mathbb{I}+\sum_{k=1}^{2}\frac{1}{2\pi i}\oint_{\partial\mathcal{U}_{\xi_{k}}}\frac{M^{(\xi_{k})}_{1}}{\sqrt{2t\epsilon_{k}f''(\xi_{k})}(\zeta-\xi_{k})\zeta}\mathrm{d}\zeta+\frac{1}{2\pi i}\oint_{\partial\mathcal{U}_{-1}}\frac{M^{(-1)}_{1}}{(\frac{15}{4}t)^{1/3}(\zeta+1)\zeta}\mathrm{d}\zeta+O(t^{-2/3})\notag\\
=&\mathbb{I}+\sum_{k=1}^{2}\frac{M^{(\xi_{k})}_{1}}{\sqrt{2t\epsilon_{k}f''(\xi_{k})}\xi_{k}}-\frac{M^{(-1)}_{1}}{(\frac{15}{4}t)^{1/3}}+O(t^{-2/3})
=\mathbb{I}-\frac{M^{(-1)}_{1}}{(\frac{15}{4}t)^{1/3}}+O(t^{-1/2}).
\end{align}

\subsubsection{Pure $\bar{\partial}$-RH problem }
In this subsection, we mainly analyse the pure $\bar{\partial}$-problem with $\bar{\partial}R^{(2)}\neq 0$. Define
\begin{align}\label{130}
M^{(5)}(z)=M^{(4)}(z)\left(M^{RHP}(z)\right)^{-1},
\end{align}
which is continuous and has no jumps in the complex plane and satisfies
the following pure $\bar{\partial}$-problem.

\begin{rhp}\label{rhp16}
Find a matrix-valued function $M^{(5)}(z)$ has the
following properties:\\

$\bullet$ $M^{(5)}(z)$  is continuous in  $\mathbb{C}\setminus (\mathbb{R}\cup\Sigma^{(4)})$;\\

$\bullet$ $\bar{\partial}M^{(5)}(z)=M^{(5)}(z)W^{(5)}(z), z\in \mathbb{C}$, where
\begin{align}\label{131}
W^{(5)}(z)=M^{RHP}(z)\bar{\partial}R^{(3)}(z)\left(M^{RHP}(z)\right)^{-1};
\end{align}

$\bullet$ $M^{(5)}(z)=\mathbb{I}+O(\frac{1}{z})$ as $z\rightarrow\infty$.
\end{rhp}

The solution of pure $\bar{\partial}$-problem can be written as
\begin{align}\label{132}
M^{(5)}(z)=\mathbb{I}-\frac{1}{\pi}\iint_{\mathbb{C}}\frac{M^{(5)}(\zeta)W^{(5)}(\zeta)}{\zeta-z}\mathrm{d}A(\zeta),
\end{align}
where $\mathrm{d}A(\zeta)$ is the Lebesgue measure, and  the operator form of  equation \eqref{132} is
\begin{align}
(\mathbb{I}-\mathcal{S})M^{(5)}(z)=\mathbb{I},\nonumber
\end{align}
where $\mathcal{S}$ is the Cauchy operator
\begin{align}
\mathcal{S}[\mathcal{F}](z)=-\frac{1}{\pi}\iint_{\mathbb{C}}\frac{\mathcal{F}(\zeta)W^{(5)}(\zeta)}{\zeta-z}\mathrm{d}A(\zeta).\nonumber
\end{align}

\begin{prop}\label{prop15}
For large time $t$,
\begin{align}
\left\Vert \mathcal{S}\right\Vert_{L^{\infty}\rightarrow L^{\infty}}\lesssim t^{-\frac{1}{6}},\nonumber
\end{align}
which states that the operator $\mathbb{I}-\mathcal{S}$ is invertible and the solution of pure $\bar{\partial}$-problem exists and is
unique.
\end{prop}

Analogously, to reconstruct the potential $q(x, t)$, we need  to discuss the long
time asymptotic behaviors of $M^{(5)}_{1}$ which is defined in the
asymptotic expansion of $M^{(5)}(z)$ as $z\rightarrow \infty$, given by
\begin{align}\label{136}
M^{(5)}(z)=\mathbb{I}+\frac{M^{(5)}_{1}}{z}+O(z^{-2}), \quad z\rightarrow \infty,
 \end{align}
\begin{align}
M^{(5)}_{1}=\frac{1}{\pi}\iint_{\mathbb{C}}M^{(5)}(\zeta)W^{(5)}(\zeta)\mathrm{d}A(\zeta).\nonumber
\end{align}
Take $z=0$ in \eqref{132}, then
\begin{align}
M^{(5)}(0)=\mathbb{I}-\frac{1}{\pi}\iint_{\mathbb{C}}\frac{M^{(5)}(\zeta)W^{(5)}(\zeta)}{\zeta}\mathrm{d}A(\zeta).\nonumber
\end{align}
The $M^{(5)}_{1},M^{(5)}(0)$ possess the following proposition.

\begin{prop}\label{prop16}
For large time $t$, $M^{(5)}_{1},M^{(5)}(0)$  admit the following inequality
\begin{align}
\left\vert M^{(5)}_{1}\right\vert \lesssim t^{-\frac{1}{2}},\qquad \left\vert M^{(5)}(0)-\mathbb{I}\right\vert \lesssim t^{-\frac{1}{2}}.\nonumber
 \end{align}
\end{prop}

\begin{prop}\label{prop16.1}
For large time $t$, $M^{(3)}(0)$  satisfies the estimate
\begin{align}\label{139}
\left\vert M^{(3)}(0)\right\vert=E(0)+t^{-\frac{1}{2}},
 \end{align}
 where $E(0)$ is given by \eqref{129}.
\end{prop}

\subsection{The final step}
Now, we are ready to give the proof of Theorem \ref{thm2} as $t\rightarrow\infty$. In terms of the transformations \eqref{41}, \eqref{102}, \eqref{107}, \eqref{113}, \eqref{130}, then the solution of RH problem \ref{rhp1} is
given by
\begin{align}\label{140}
M(z)=(I+\frac{\sigma_{1}}{z}M^{(3)}(0)^{-1})M^{(5)}(z)E(z)\delta^{-\sigma_{3}}(z)+O(e^{-ct}).
 \end{align}
Considering that \eqref{126},\eqref{128},\eqref{129},\eqref{136},\eqref{139},  the reconstruction formula \eqref{27} arrives at
\begin{align}\label{141}
q(x,t)&=-i\left(-i-i\frac{[M^{(-1)}_{1}]_{22}}{(\frac{15}{4}t)^{1/3}}-\frac{[M^{(-1)}_{1}]_{12}}{(\frac{15}{4}t)^{1/3}}\right)+O(t^{-\frac{1}{2}})\notag\\
&=-1-\frac{i}{2}(\frac{15}{4}t)^{-\frac{1}{3}}\left(\int_{s}^{\infty}u^{2}(\zeta)d\zeta+e^{i\varphi_{0}}u(s)\right)+O(t^{-\frac{1}{2}}).
 \end{align}

\begin{appendices}
\section{}\label{A}
In Appendices \ref{A}, we aim to solve the model RH problem $[M_{1}^{(\xi_{1})}]_{12}$ explicitly by
introducing the following transformation(see Fig. 4)
\begin{align}
M^{mod}=M^{(\xi_{1})}G_{j},\qquad  s \in \Omega_{j},\qquad j=0, \cdots, 4,\nonumber
\end{align}
where
\begin{gather}
G_{0}=e^{\frac{1}{4}is^{2}\sigma_{3}}s^{-i\nu(\xi_{1})\sigma_{3}},\notag\\
G_{1}=G_{0}\left(\begin{array}{cc}
    1  &  -\frac{\overline{r_{\xi_{1}}}}{1-\left\vert r_{\xi_{1}}\right\vert^{2}}\\
     0  &  1\\
\end{array}\right),\quad G_{2}=G_{0}\left(\begin{array}{cc}
    1  &  0\\
     -\frac{r_{\xi_{1}}}{1-\left\vert r_{\xi_{1}}\right\vert^{2}}  &  1\\
\end{array}\right),\notag\\
G_{3}=G_{0}\left(\begin{array}{cc}
    1  &  \overline{r_{\xi_{1}}}\\
    0  &  1\\
\end{array}\right),\quad G_{4}=G_{0}\left(\begin{array}{cc}
    1  &  0\\
    r_{\xi_{1}}  &  1\\
\end{array}\right).\nonumber
\end{gather}
Through this transformation, we construct a model RH problem for $M^{mod}$ with a constant jump matrix.
\begin{rhp}\label{rhp9}
 Find an analysis function $M^{mod}(s)$ with the
following properties:\\

 $\bullet$ $M^{mod}(s)$ is meromorphic in $\mathbb{C}\setminus \mathbb{R}$;\\

 $\bullet$ $M_{+}^{mod}(s)=M_{-}^{mod}(s)J^{mod}(\xi_{1}), s\in \mathbb{R}$, where
\begin{align}
J^{mod}(\xi_{1})=\left(\begin{array}{cc}
    1-\left\vert r_{\xi_{1}}\right\vert^{2}  &   -\overline{r_{\xi_{1}}}\\
     r_{\xi_{1}} &  1\\
\end{array}\right);\nonumber
\end{align}

$\bullet$ $M^{mod}(s)\rightarrow e^{\frac{1}{4}is^{2}\sigma_{3}}s^{-i\nu(\xi_{1})\sigma_{3}}$ as $s\rightarrow\infty$.
\end{rhp}

According to the Liouville's theorem and parabolic cylinder functions, this RH problem can be solved explicitly. It is not hard to find that $\frac{\mathrm{d}}{\mathrm{d}s}M^{mod}(M^{mod})^{-1}$ possesses continuous jump along any of the rays,  and then it admits
\begin{align}
\frac{\mathrm{d}}{\mathrm{d}s}M^{mod}+\left(\begin{array}{cc}
    -\frac{i}{2}s  &  \psi_{1}\\
  \psi_{2} &  \frac{i}{2}s\\
\end{array}\right)M^{mod}=0,\nonumber
\end{align}
where $\psi_{1}=i\left[M_{1}^{(\xi_{1})}\right]_{12}, \psi_{2}=-i\left[M_{1}^{(\xi_{1})}\right]_{21}$.
It can be solved as
\begin{align}
M^{mod}=\left(\begin{array}{cc}
    M_{11}^{mod}  &  \frac{\frac{i}{2}s M_{22}^{mod}+\frac{\mathrm{d}M_{22}^{mod}}{\mathrm{d}s}}{-\psi_{2}}\\
  \frac{-\frac{i}{2}s M_{11}^{mod}+\frac{\mathrm{d}M_{11}^{mod}}{\mathrm{d}s}}{-\psi_{1}} &  M_{22}^{mod}\\
\end{array}\right),\nonumber
\end{align}
where the functions $M_{jj}^{mod}, j=1, 2,$ satisfy  the following equations
\begin{align}
&\frac{\mathrm{d}^{2}}{\mathrm{d}s^{2}}M_{11}^{mod}+(-\frac{i}{2}-\psi_{1}\psi_{2}+\frac{s^{2}}{4})M_{11}^{mod}=0,\notag\\
&\frac{\mathrm{d}^{2}}{\mathrm{d}s^{2}}M_{22}^{mod}+(\frac{i}{2}-\psi_{1}\psi_{2}+\frac{s^{2}}{4})M_{22}^{mod}=0.\nonumber
\end{align}
 Due to the above equations are standard parabolic cylinder equation and $M_{11}^{mod}\rightarrow e^{\frac{1}{4}is^{2}}s^{-i\nu(\xi_{1})},$ $ M_{22}^{mod}\rightarrow e^{-\frac{1}{4}is^{2}}s^{i\nu(\xi_{1})}, s\rightarrow\infty$, one has
\begin{align}
M_{11}^{mod}=\left\{
\begin{array}{lr}
(e^{\frac{3\pi}{4}i})^{i\nu(\xi_{1})}D_{-i\nu(\xi_{1})}(e^{\frac{3\pi}{4}i}s) \qquad \mbox{Im}(s)<0,\\
\\
(e^{-\frac{\pi}{4}i})^{i\nu(\xi_{1})}D_{-i\nu(\xi_{1})}(e^{-\frac{\pi}{4}i}s) \qquad \qquad \mbox{Im}(s)>0,
  \end{array}
\right.\nonumber
\end{align}
\begin{align}
M_{22}^{mod}=\left\{
\begin{array}{lr}
(e^{\frac{\pi}{4}i})^{-i\nu(\xi_{1})}D_{i\nu(\xi_{1})}(e^{\frac{\pi}{4}i}s) \ \qquad \mbox{Im}(s)<0,\\
\\
(e^{-\frac{3\pi}{4}i})^{-i\nu(\xi_{1})}D_{i\nu(\xi_{1})}(e^{-\frac{3\pi}{4}i}s) \quad \qquad \mbox{Im}(s)>0.
  \end{array}
\right.\nonumber
\end{align}
Then, we obtain
\begin{gather}
M^{mod}_{-}(s)^{-1}M^{mod}_{+}(s)=M^{mod}_{-}(0)^{-1}M^{mod}_{+}(0)=\notag\\
\left(\begin{array}{cc}
  e^{-\frac{3}{4}\pi\nu}\frac{2^{\frac{-i\nu}{2}}\sqrt{\pi}}{\Gamma(\frac{1+i\nu}{2})}    & e^{\frac{\pi}{4}i}e^{\frac{\pi\nu}{4}}\frac{2^{\frac{1+i\nu}{2}}\sqrt{\pi}}{\psi_{2}\Gamma(\frac{-i\nu}{2})} \\
  e^{\frac{3\pi}{4}i}e^{-\frac{3}{4}\pi\nu}\frac{2^{\frac{1-i\nu}{2}}\sqrt{\pi}}{\psi_{1}\Gamma(\frac{i\nu}{2})} &   e^{\frac{\pi\nu}{4}}\frac{2^{\frac{i\nu}{2}}\sqrt{\pi}}{\Gamma(\frac{1-i\nu}{2})}\\
\end{array}\right)^{-1}\notag\\
\left(\begin{array}{cc}
  e^{\frac{\pi\nu}{4}}\frac{2^{\frac{-i\nu}{2}}\sqrt{\pi}}{\Gamma(\frac{1+i\nu}{2})}    & e^{-\frac{3\pi}{4}i}e^{-\frac{3}{4}\pi\nu}\frac{2^{\frac{1+i\nu}{2}}\sqrt{\pi}}{\psi_{2}\Gamma(\frac{-i\nu}{2})} \\
  e^{-\frac{\pi}{4}i}e^{\frac{\pi}{4}\nu}\frac{2^{\frac{1-i\nu}{2}}\sqrt{\pi}}{\psi_{1}\Gamma(\frac{i\nu}{2})} &   e^{-\frac{3}{4}\pi\nu}\frac{2^{\frac{i\nu}{2}}\sqrt{\pi}}{\Gamma(\frac{1-i\nu}{2})}\\
\end{array}\right)\notag\\
=\left(\begin{array}{cc}
    1-\left\vert r_{\xi_{1}}\right\vert^{2}  &   -\overline{r_{\xi_{1}}}\\
     r_{\xi_{1}} &  1\\
\end{array}\right),\nonumber
\end{gather}
which leads to
\begin{align}\label{A9}
[M_{1}^{(\xi_{1})}]_{12}=\frac{\sqrt{2\pi}e^{-\frac{\pi}{4}i-\frac{\pi\nu(\xi_{1})}{2}}}{ir_{\xi_{1}}\Gamma(i\nu(\xi_{1}))}.
\end{align}
Carrying out similar procedures, we easily derive
\begin{align}\label{A10}
[M_{1}^{(\xi_{2})}]_{12}=\frac{\sqrt{2\pi}e^{\frac{\pi}{4}i-\frac{\pi\nu(\xi_{2})}{2}}}{ir_{\xi_{2}}\Gamma(-i\nu(\xi_{2}))},\
[M_{1}^{(\xi_{3})}]_{12}=\frac{\sqrt{2\pi}e^{-\frac{\pi}{4}i-\frac{\pi\nu(\xi_{3})}{2}}}{ir_{\xi_{3}}\Gamma(i\nu(\xi_{3}))},\
[M_{1}^{(\xi_{4})}]_{12}=\frac{\sqrt{2\pi}e^{\frac{\pi}{4}i-\frac{\pi\nu(\xi_{4})}{2}}}{ir_{\xi_{4}}\Gamma(-i\nu(\xi_{4}))}.
\end{align}
\section{}\label{B}
Here we transform the RH problem \ref{rhp14} into a standard Painlev\'{e} II equation via suitable deformation. Firstly, we add four new auxiliary wires $X_{j}, j = 1, 2, 3, 4$ going through  the point $\hat{k}=0$
at the  $\frac{\pi}{3}$ angle with real axis. Then the complex
plane is divided into eight regions $\Omega_{j}, j=1, \cdots , 8$, See Fig. 9.\\
\centerline{\begin{tikzpicture}
\draw[->][dashed,thick](0,0)--(0.5,0.85);
\draw[-][dashed,thick](0.5,0.85)--(1,1.7);
\draw[->][dashed,thick](0,0)--(0.5,-0.85);
\draw[-][dashed,thick](0.5,-0.85)--(1,-1.7);
\draw[-][dashed,thick](0,0)--(-0.5,0.85);
\draw[<-][dashed,thick](-0.5,0.85)--(-1,1.7);
\draw[-][dashed,thick](0,0)--(-0.5,-0.85);
\draw[<-][dashed,thick](-0.5,-0.85)--(-1,-1.7);
\draw[-][dashed,thick](-3,0)--(-1,0);
\draw[-][dashed,thick](1,0)--(3,0);
\draw[-][thick](-1,0)--(0,0) circle [radius=0.035]node[below]{$0$};
\draw[->][thick](0,0)--(0.5,0);
\draw[-][thick](0.5,0)--(1,0);
\draw[-][thick](-2,1)--(-1,0);
\draw[->][thick](-3,2)--(-2,1);
\draw[-][thick](-2,-1)--(-1,0);
\draw[->][thick](-3,-2)--(-2,-1);
\draw[<-][thick](2,1)--(1,0);
\draw[-][thick](3,2)--(2,1);
\draw[<-][thick](2,-1)--(1,0);
\draw[-][thick](3,-2)--(2,-1);
\draw[fill] (-1,0) node[below]{$\hat{k}_{4}$};
\draw[fill] (1,0) node[below]{$\hat{k}_{3}$};
\draw[fill] [cyan](-1,1.3) node{$X_{2}$};
\draw[fill] [cyan](-1,-1.4) node{$X_{3}$};
\draw[fill] [cyan](1,1.3) node{$X_{1}$};
\draw[fill] [cyan](1,-1.4) node{$X_{4}$};
\draw[fill] [red](-1.5,1) node{$\Omega_{3}$};
\draw[fill] [red](-1.5,-1) node{$\Omega_{4}$};
\draw[fill] [red](1.5,1) node{$\Omega_{1}$};
\draw[fill] [red](1.5,-1) node{$\Omega_{6}$};
\draw[fill] [red](0,1) node{$\Omega_{2}$};
\draw[fill] [red](0,-1) node{$\Omega_{5}$};
\end{tikzpicture}}
\centerline{\noindent {\small \textbf{Figure 9.} (Color online) The deformation of jump contour between $\tilde{M}^{(0)}$ and $M^{(0)}$.}}

Via defining
\begin{align}
P(\hat{k})=
\begin{cases}
b_{+}^{-1},\qquad \hat{k}\in \Omega_{1}\cup\Omega_{3},\\
b_{-},\qquad \hat{k}\in \Omega_{4}\cup\Omega_{6},\\
\mathbb{I}, \qquad \mbox{others}
\end{cases}\nonumber
\end{align}
and performing a transformation
\begin{align}\label{B2}
\tilde{M}^{(0)}(\hat{k})=M^{(0)}(\hat{k})P(\hat{k}),
\end{align}
we generate the following RH problem.
\begin{rhp}\label{rhpB1}
The analysis function $M^{(0)}(\hat{k})$ has the
following properties:\\

 $\bullet$ $\tilde{M}^{(0)}(\hat{k})$ is meromorphic in $\mathcal{U}_{0}\setminus \Sigma_{P}$, where $\Sigma_{P}=\bigcup_{j=1}^{4}X_{j}$;\\

 $\bullet$ $\tilde{M}_{+}^{(0)}(\hat{k})=\tilde{M}^{(0)}_{-}(\hat{k})\tilde{J}^{(0)}(\hat{k}), \hat{k}\in \Sigma_{P}$, where
\begin{align}
\tilde{J}^{(0)}(\hat{k})=
\begin{cases}
b_{+},\qquad \hat{k}\in X_{1}\cup X_{2},\\
b_{-},\qquad \hat{k}\in X_{3}\cup X_{4},\\
\mathbb{I}, \qquad \mbox{others};
\end{cases}\nonumber
\end{align}

$\bullet$ $\tilde{M}^{(0)}(\hat{k})\rightarrow\mathbb{I}$ as $\hat{k}\rightarrow\infty$.
\end{rhp}
Let $\varphi_{0}=arg R(-1),$ so $R(-1)=\vert R(-1)\vert e^{i\varphi_{0}}$. Then, we find
\begin{align}\label{B4}
\tilde{M}^{(0)}(\hat{k})=\sigma_{1}e^{-\frac{i\varphi_{0}}{2}\hat{\sigma}_{3}}M^{P}(\hat{k})\sigma_{1},
\end{align}
where $M^{P}(\hat{k})$ becomes a standard Painlev\'{e} II model, whose solution  is given by
\begin{align}
M^{P}(\hat{k})=I+\frac{M_{1}^{P}(s)}{\hat{k}}+O(\hat{k}^{-2}),\nonumber
\end{align}
where
\begin{align}
M_{1}^{P}(s)=\frac{1}{2}\left(\begin{array}{cc}
    -i\int_{s}^{\infty}u^{2}(\zeta)d\zeta  &   u(s)\\
     u(s) &  i\int_{s}^{\infty}u^{2}(\zeta)d\zeta\\
\end{array}\right),\nonumber
\end{align}
and $u(s)$ is a solution of the Painlev\'{e} II equation
\begin{align}
u_{ss}=2u^{3}+su,\qquad s\in \mathbb{R}.\nonumber
\end{align}
Considering the transformations \eqref{B2}, \eqref{B4}, and expanding $M^{(0)}(\hat{k})$ along the region
$\Omega_{2}$ or $\Omega_{5}$, we have
\begin{align}
M^{(0)}(\hat{k})=I+\frac{M_{1}^{(0)}(s)}{\hat{k}}+O(\hat{k}^{-2}),\nonumber
\end{align}
where
\begin{align}
M_{1}^{(0)}(s)=\frac{i}{2}\left(\begin{array}{cc}
    -\int_{s}^{\infty}u^{2}(\zeta)d\zeta  &   e^{i\varphi_{0}}u(s)\\
     e^{-i\varphi_{0}}u(s) &  \int_{s}^{\infty}u^{2}(\zeta)d\zeta\\
\end{array}\right).\nonumber
\end{align}
According to Proposition \ref{prop13},one has
\begin{align}
M_{1}^{(-1)}(s)=M_{1}^{(0)}(s)+O(t^{-1/6})=\frac{i}{2}\left(\begin{array}{cc}
    -\int_{s}^{\infty}u^{2}(\zeta)d\zeta  &   e^{i\varphi_{0}}u(s)\\
     e^{-i\varphi_{0}}u(s) &  \int_{s}^{\infty}u^{2}(\zeta)d\zeta\\
\end{array}\right)+O(t^{-1/6}).\nonumber
\end{align}
\end{appendices}

\section*{Acknowledgements}
\hspace{0.3cm}
This work was supported by the National Natural Science Foundation of China (No. 12175069 and No. 12235007), Science and Technology Commission of Shanghai Municipality (No. 21JC1402500 and No. 22DZ2229014) and Natural Science Foundation of Shanghai (No. 23ZR1418100).

\end{document}